\documentclass[acmsmall,nonacm]{acmart}

\setcopyright{acmcopyright}
\copyrightyear{2020}
\acmYear{2020}
\acmDOI{xxxxx}

\acmConference[Conference '20]{Conference '20}{Date}{Location}

\usepackage{algpseudocode}
\algdef{SE}[SUBALG]{Indent}{EndIndent}{}{\algorithmicend\ }%
\algtext*{Indent}
\algtext*{EndIndent}
\usepackage{algorithm}

\usepackage{enumerate}

\usepackage{pictexwd,dcpic}
\usepackage{listings}

\newcommand{\mc}[1]{\ensuremath{\mathcal{#1}}}

\newcommand{\mb}[1]{\ensuremath{\mathbb{#1}}}
\newcommand{\msf}[1]{\ensuremath{\mathsf{#1}}}
\newcommand{\ra}{\rightarrow}
\newcommand{\pra}{\rightharpoonup}
\newcommand{\act}{\triangleright}
\newcommand{\N}{\mb{N}}
\newcommand{\Z}{\mb{Z}}

\newcommand{\R}{\mb{R}}

\lstdefinelanguage{icfppseudo}%
{morekeywords={abstract,case,catch,char,class,%
    def,else,extends,final,%
    if,import,%
    match,module,new,null,object,override,package,private,protected,%
    public,return,super,this,throw,trait,try,type,val,var,with,implicit,%
    macro,sealed,%
  },%
  sensitive,%
  morecomment=[l]//,%
  morestring=[b]",%
  morestring=[b]',%
  showstringspaces=false%
}[keywords,comments,strings]%

\theoremstyle{plain}
\newtheorem{mythm}{Theorem}[section]

\newtheorem{myprop}[mythm]{Proposition}

\theoremstyle{definition}
\newtheorem{mydef}[mythm]{Definition}
\newtheorem{myex}[mythm]{Example}
\newtheorem{myrmk}[mythm]{Remark}
\newtheorem*{mynot}{Notation}

\begin{document}

\title[Composing and Decomposing Op-Based CRDTs with Semidirect Products]{Composing and Decomposing Op-Based CRDTs with Semidirect Products}

\author{Matthew Weidner}
\email{maweidne@andrew.cmu.edu}
\affiliation{%
  \institution{Carnegie Mellon University}
  \streetaddress{5000 Forbes Avenue}
  \city{Pittsburgh}
  \state{PA}
  \postcode{15213}
}

\author{Heather Miller}
\email{heather.miller@cs.cmu.edu}
\affiliation{%
  \institution{Carnegie Mellon University}
  \streetaddress{5000 Forbes Avenue}
  \city{Pittsburgh}
  \state{PA}
  \postcode{15213}
}

\author{Christopher Meiklejohn}
\email{cmeiklej@andrew.cmu.edu}
\affiliation{%
  \institution{Carnegie Mellon University}
  \streetaddress{5000 Forbes Avenue}
  \city{Pittsburgh}
  \state{PA}
  \postcode{15213}
}

\begin{abstract}
Operation-based Conflict-free Replicated Data Types (CRDTs) are eventually consistent replicated data types that automatically resolve conflicts between concurrent operations.  Op-based CRDTs must be designed differently for each data type, and current designs use ad-hoc techniques to handle concurrent operations that do not naturally commute.  We present a new construction, the \textit{semidirect product} of op-based CRDTs, which combines the operations of two CRDTs into one while handling conflicts between their concurrent operations in a uniform way.  We demonstrate the construction's utility by using it to construct novel CRDTs, as well as decomposing several existing CRDTs as semidirect products of simpler CRDTs.  Although it reproduces common CRDT semantics, the semidirect product can be viewed as a restricted kind of operational transformation, thus forming a bridge between these two opposing techniques for constructing replicated data types.
\end{abstract}

\ccsdesc[500]{Software and its engineering~Data types and structures}
\ccsdesc[300]{Theory of computation~Distributed algorithms}

\keywords{CRDTs, Operational Transformation, Eventual Consistency}

\maketitle

\section{Introduction}
Geo-replication of data is a technique used in many distributed applications, such as distributed databases and collaborative document editing programs, to reduce user-perceived latency and increase fault tolerance.  However, geo-replication is challenging: in order to achieve fault tolerance and high performance, replicas need to avoid the large communication costs of coordinating over geographically distributed networks.  This is typically achieved by allowing users to update their own replicas locally, then asynchronously propagate updates to other replicas in the background.  This poses an interesting concurrency challenge: what should the final outcome be in the face of conflicting concurrent operations by different replicas, once all of the updates are delivered to all replicas?

\textit{Conflict-free Replicated Data Types (CRDTs)} \cite{crdt_survey_2011, crdt_summary_2018} are a class of highly available replicated data types that provide a principled solution to this problem.  They allow developers to write applications using ordinary sequential data type operations acting on replicated state.  Updates to the replicated state are propagated asynchronously, with conflicts between concurrent operations resolved automatically using type-specific rules.  CRDTs have been deployed in highly available geo-replicated data storage systems such as Antidote \cite{antidote} and Riak \cite{riak_datatypes}.

Operation-based (op-based) CRDTs \cite{crdt_survey_2011, crdt_summary_2018} are a type of CRDT that function by converting sequential data type operations into messages, which are then broadcast to all replicas in \textit{causal order}, a natural partial order that can be enforced without coordination.  Replicas apply received messages to their state in such a way that they all end up in the same state, even if they receive concurrent messages in different orders.  A simple example is a counter supporting increment and decrement operations: these operations naturally commute, so they can be directly broadcast to all replicas, which then apply the operations in the order they receive them.  More generally, any commutative data type is trivially an op-based CRDT.

However, op-based CRDTs for non-commutative data types must explicitly handle conflicts between non-commuting concurrent messages.  For example, a set CRDT must decide what to do in the face of operations that concurrently add and remove an element to the set, while ensuring that all replicas eventually reach the same result.  Current designs handle such conflicts using ad-hoc techniques that differ for each data type.  It is thus difficult to design CRDTs for new data types or to add operations to existing CRDTs.

For example, suppose we are making a Slack-like chat application containing multiple ``channels'' that users can join.  We can store the map from channel names to the set of users in each channel using a CRDT map (dictionary) \cite{antidote, riak_datatypes}, with CRDT sets as the values, as shown in Figure \ref{fig:pseudo}.  In this example, there are two replicas, replica $A$ and replica $B$, stored at two different nodes, $A$ and $B$, that are executing the following application code.

\begin{figure}[!htb]
    \centering
    \begin{minipage}{.5\textwidth}
        \centering
        \textbf{\texttt{Node A}}
\begin{lstlisting}
// create the channels CRDT 
channels: CRDTMap[String, Set[String]] = 
  initReplicated(
    "general" -> ["alice", "bob"], 
    "random"-> ["alice", "charlie"]
  )

// replicate channels on node B
nodeB ! channels

// create a new channel "memes" in 
// channels and add alice
(m, channels) = 
  apply("memes", add("alice"), channels)

// send apply to node B
nodeB ! m
\end{lstlisting}
    \end{minipage}%
    \begin{minipage}{0.5\textwidth}
        \centering
        \textbf{\texttt{Node B}}
        \scriptsize
\begin{lstlisting}
// add dave to every channel using 
// higher-order map
(m, channels) = 
  homap(add("dave"), channels)

// send homap to node A
nodeA ! m
\end{lstlisting}
        \normalsize        
        \hfill
    \end{minipage}
    
\caption{Example program showing distributed use of the replicated value \texttt{channels} in a Slack-like chat application by two nodes, $A$ and $B$.  $\msf{apply}(k, \msf{add}(v))$ adds $v$ to the set at key $k$, initializing the value at $k$ if necessary, while $\msf{homap}(\msf{add}(v))$ adds $v$ to every set in the map.  Unfortunately, these two operations do not commute, so after receiving each other's messages, $A$ and $B$ end up with different values for \texttt{channels} (see Figure \ref{fig:dictionaries}).}
\label{fig:pseudo}
\end{figure}

\begin{figure}[!htb]
\centering
\includegraphics[width=0.95\textwidth]{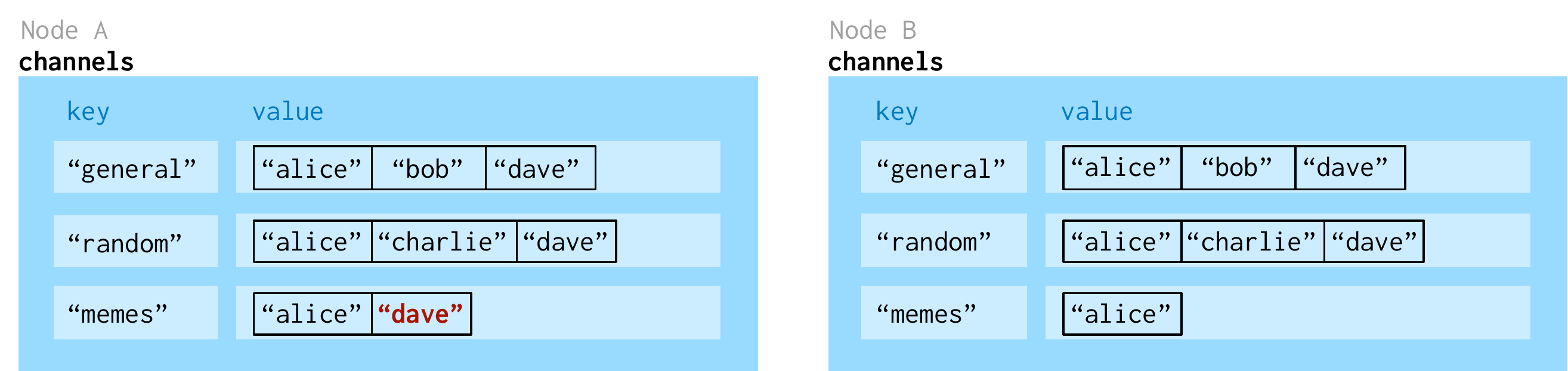}
\caption{Replicas $A$ and $B$ after executing the example program shown in Figure~\ref{fig:pseudo}. Because the operations from Figure~\ref{fig:pseudo} do not commute, $A$ and $B$ end up with different values for \texttt{channels}.}
\label{fig:dictionaries}
\end{figure}

Suppose we wish to add a higher-order map operation $\msf{homap}$, which applies a function to every value currently present in the map. This operation could be used to add a user \texttt{dave} to every channel at node $B$.  In Figure \ref{fig:pseudo}, we run into the issue that this operation does not commute with operations that initialize the value at a key: node $B$ processes the higher-order map operation before the operation that initializes channel \texttt{memes}, so it does not add \texttt{dave} to \texttt{memes}, while node $A$ processes the operations in the opposite order, so it does add \texttt{dave} to \texttt{memes}, as shown in Figure~\ref{fig:dictionaries}.   Hence it is not obvious how to add such an operation to a map CRDT, and as a result, existing map designs do not include one.

Our goal in this paper is to make it easy to add operations like this to CRDTs, even when they do not commute with existing operations.  More generally, we aim to \textit{compose} op-based CRDTs, combining their sets of operations while handling conflicts between them in a uniform way.  One can then add operations to an existing CRDT by composing it with another CRDT that only implements the new operations.

We do so by presenting a new CRDT construction technique, the \textit{semidirect product of op-based CRDTs}, which composes two op-based CRDTs into a new CRDT with both of their operations.  We demonstrate through numerous examples that the semidirect product can be used both to create novel CRDTs through composition and to decompose existing CRDTs into simpler components.  Our novel CRDT designs include maps with a higher-order map operation, sequence types with reverse or range removal operations, and an integer register supporting addition and multiplication operations.  These examples widen the range of data types available as CRDTs without substantial design effort.  The existing CRDTs we decompose include sets, flags, and resettable CRDTs.  These examples help explain seemingly ad-hoc CRDT designs by decomposing them into simpler parts, many of which are commutative data types, with the semidirect product handling conflicts between non-commuting concurrent operations in a uniform way.

We further demonstrate the semidirect product's generality by giving a criterion for when an op-based CRDT designed using an existing general model \cite{pure_op_based_crdts_extended} can be decomposed as a semidirect product of simpler CRDTs, and we show that it applies in many cases.

Briefly, the semidirect product works as follows.  Conflicts between non-commuting concurrent operations from the two CRDTs are handled according to an \textit{arbitration order}, which specifies that operations from the first CRDT should be applied before concurrent operations from the second CRDT.  Since directly reordering operations in this way is not always defined, we instead make use of a transformation function, which transforms operations from the first CRDT to take into account concurrent operations from the second CRDT.  Our construction can thus be viewed as a restricted kind of operational transformation \cite{ot_ressel}, an alternative method of constructing replicated data types that is often contrasted with CRDTs, as we discuss in Section \ref{sec:ot}.  In comparison with general operational transformation, the semidirect product has reduced complexity, thus avoiding CRDT proponents' main criticism of operational transformation.

The semidirect product of CRDTs is named after the semidirect product of groups, a construction from abstract algebra that inspired their design, as described in the appendix.

A shortened version of this paper will appear at the PaPoC Workshop 2020 \cite{papoc}.

\section{Background on Op-Based CRDTs}
\label{sec:background}
\begin{mynot}
We use $f: S \pra T$ to denote that $f$ is a partial function from $S$ to $T$, i.e., a function defined on a subset of $S$.  We write $f(s) = \bot$ to indicate that $f$ is not defined on $s$.
\end{mynot}

As mentioned above, an op-based CRDT is a replicated data type in which replicas convert sequential data type operations into messages that they broadcast to other replicas.  Replicas receive these messages and apply them to their states in causal order (defined below), in such a way that concurrent messages commute.
\begin{mydef}
The \textit{causal order} is the partial order $\prec$ on messages defined by the transitive closure of the rule: $m_1 \prec m_2$ if the replica that sent $m_2$ did so after receiving $m_1$, or if $m_1$ and $m_2$ were sent by the same replica and $m_1$ was sent before $m_2$.  Two messages $m_1, m_2$ are \textit{concurrent} if $m_1 \nprec m_2$ and $m_2 \nprec m_1$.  The requirement that replicas receive messages in causal order means that a replica should not receive a message $m_2$ until after it has received all messages $m_1 \prec m_2$.
\end{mydef}

A simple example is the op-based counter CRDT.  This has state space $\Z$ with operations $\msf{add}(n)$ for $n \in \Z$, acting as $\msf{add}(n): \sigma \mapsto n + \sigma$.  These operations naturally commute.

More complicated CRDTs, such as set CRDTs, attach extra metadata to states, to handle the fact that their operations do not naturally commute.  Additionally, instead of sending operations directly in their messages to other replicas, they may modify them or attach metadata.

Formally, we adopt the following definition of an op-based CRDT, based on that of Shapiro et~al.\ \cite{crdt_survey_2011} but with notation more similar to \cite[\S 3]{pure_op_based_crdts_extended}.
\begin{mydef}
\label{def:crdt}
An \emph{op-based CRDT} is a tuple $(\Sigma, \sigma^0, \msf{prepare}, \msf{effect}, \msf{eval})$ of the form given in Figure \ref{fig:op-based}, such that:
\begin{enumerate}[(i)]
   \item For all $m = \msf{prepare}(o, \sigma, r)$, $m \cdot \sigma \neq \bot$
   \item For all $\sigma \in \Sigma$ and all messages $m_1, m_2$, if $\sigma$ can appear as a replica state and $m_1, m_2$ can appear as concurrent messages in an execution of Algorithm \ref{alg:crdt_use}, and $m_1 \cdot \sigma \neq \bot$ and $m_2 \cdot \sigma \neq \bot$, then
    \[
    m_1 \cdot (m_2 \cdot \sigma) = m_2 \cdot (m_1 \cdot \sigma) \neq \bot.
    \]
\end{enumerate}
\end{mydef}

\begin{figure}
\begin{tabular}{ccp{4.9cm}}
$\Sigma$ & : & Set of states \\
$\msf{prepare}(o, \sigma, r)$ & : & Prepares a \textit{message} $m$ given an operation $o$ by replica $r$ in state $\sigma$ \\
$\msf{effect}(m, \sigma)$ & : & Partial function that applies a prepared message $m$ to a state $\sigma$, returning the resulting state \\
$m \cdot \sigma$ & : & Abbreviation for $\msf{effect}(m, \sigma)$ when the CRDT is clear from context \\
$\msf{eval}(q, \sigma)$ & : & Read-only evaluation of a query $q$ on a state $\sigma$
\end{tabular}
\caption{Components of an op-based CRDT.}
\label{fig:op-based}
\end{figure}

The definition ensures that operations commute if they could be issued concurrently.

\begin{myex}
We can formalize the op-based counter CRDT as $\Sigma = \Z$, $\msf{prepare}(\msf{add}(n), \sigma, r) = \msf{add}(n)$, $\msf{effect}(\msf{add}(n), \sigma) = n + \sigma$, and $\msf{eval}(\msf{value}, \sigma) = \sigma$.
\end{myex}


Algorithm \ref{alg:crdt_use} formalizes the use of an op-based CRDT by a group of replicas.  Initially, all replicas are in the initial state $\sigma^0$.  At any time, a replica can issue an operation $o$, causing a message $m$ to be prepared and broadcast to all replicas.  All replicas apply received messages to their state in causal order using $\msf{effect}$.  Replicas can also be queried to return external information about their state.  For instance, a set CRDT could have a query to return the elements of the set based on the internal metadata-enhanced state $\sigma$.

\begin{algorithm}[ht!]
\begin{algorithmic}[1]
\State \textbf{state} $\sigma \in \Sigma$, initially $\sigma^0$
\State \textbf{on} $\msf{operation}(o)$:
\Indent
  \State $m \gets \msf{prepare}(o, \sigma, r)$
  \State $\sigma \gets \msf{effect}(m, \sigma)$
  \State Broadcast $m$ to other replicas
\EndIndent
\State \textbf{on} $\msf{receive}(m)$:
\Indent
  \State $\sigma \gets \msf{effect}(m, \sigma)$
\EndIndent
  \State \textbf{on} $\msf{query}(q)$:
\Indent
  \State \Return $\msf{eval}(q, \sigma)$
\EndIndent
\end{algorithmic}
\caption{Distributed algorithm describing the use of an op-based CRDT by a replica $r$, based on \cite[Algorithm 1]{pure_op_based_crdts_extended} \cite{crdt_survey_2011}.  Messages are assumed to be received in causal order.
}
\label{alg:crdt_use}
\end{algorithm}

Properties (i) and (ii) of an op-based CRDT imply \textit{eventual consistency}: two replicas that have received the same messages end up in the same state $\sigma$ when using Algorithm \ref{alg:crdt_use}, even if they receive concurrent messages in different orders \cite[Proposition 2.2]{crdt_survey_2011}.

The main advantage of op-based CRDTs over strongly consistent data types, like data in traditional databases, is that causally ordered message delivery can be enforced without coordination between replicas.  Thus replicas can immediately apply changes to their own copies of the data, and they can send messages without requiring a costly consensus protocol to put all messages in a consistent total order \cite{total_order_survey}.

\begin{myrmk}
One way to enforce causally ordered delivery in an implementation of Algorithm \ref{alg:crdt_use} is to use \textit{vector clocks} \cite{fidge, mattern}, which are functions from replica ids to $\N$.  Each replica maintains a vector clock $t$ such that $t(r)$ is the number of messages it has received from replica $r$.  When sending a message $m$, a replica increments its own entry in $t$ and attaches a copy of $t$ to $m$ as metadata (its \textit{timestamp}).  Thus $m_1 \preceq m_2$ if and only if their corresponding timestamps $t_1, t_2$ satisfy $t_1 \leq t_2$, and they are concurrent if and only if $t_1 \nleq t_2$ and $t_2 \nleq t_1$.  To enforce causally ordered delivery, we wait to deliver a message $m$ from a replica $r_1$ with timestamp $t$ to another replica $r_2$ until $r_2$'s vector clock $t_2$ satisfies: $t_2(r_1) = t(r_1) - 1$ and for all $r \neq r_1$, $t_2(r) \ge t(r)$.  Besides enforcing causally ordered delivery, some CRDT algorithms explicitly include vector clocks in messages so that $\msf{effect}$ can query the causal order.  Our semidirect product construction does this to determine when messages are concurrent.
\end{myrmk}


\section{Semidirect Products}

\subsection{Motivating Example}
\label{sec:example}
To motivate the semidirect product construction, suppose we wish to construct an integer register CRDT supporting addition and multiplication operations.  Specifically, for $n \in \Z$, we want to allow operations $\msf{add}(n): \sigma \mapsto n + \sigma$ and $\msf{mult}(n): \sigma \mapsto n \times \sigma$.  The $\msf{add}$ operations alone form a CRDT because they naturally commute, and likewise for the $\msf{mult}$ operations, but they do not commute with each other.  Because of this, such a CRDT was not previously known, despite the simplicity of its interface.

As a first attempt, let us dictate that in the face of concurrent $\msf{add}$ and $\msf{mult}$ operations, a replica applies all of the $\msf{add}$ operations first, followed by all of the $\msf{mult}$ operations.  For example, starting in state 1, if two replicas concurrently issue operations $\msf{add}(1)$ and $\msf{mult}(3)$, then regardless of the order in which they receive these operations, all replicas compute the final state as $\msf{mult}(3) \cdot (\msf{add}(1) \cdot 1) = 6$.  We call this order of operations the \textit{arbitration order}.  Non-concurrent operations should continue to be applied in the order they were generated, i.e., in causal order.

Unfortunately, this approach is not well-defined in general.  For example, consider the following scenario:
\[
\begindc{\commdiag}[50]
\obj(5,7){$A$}
\obj(5,00){$B$}
\obj(10,7)[A1]{1}
\obj(20,7)[A2]{2}
\obj(30,7)[A3]{3}
\obj(40,7)[A4]{?}
\obj(50,7)[A5]{?}
\obj(10,00)[B1]{1}
\obj(20,00)[B2]{3}
\obj(30,00)[B3]{7}
\obj(40,00)[B4]{?}
\obj(50,00)[B5]{?}

\mor{A1}{A2}{$\msf{mult}(2)$}
\mor{A2}{A3}{$\msf{add}(1)$}
\mor{A3}{A4}{}[\atleft, \dotArrow]
\mor{A4}{A5}{}[\atleft, \dotArrow]
\mor{B1}{B2}{$\msf{mult}(3)$}[\atright, \solidarrow]
\mor{B2}{B3}{$\msf{add}(4)$}[\atright, \solidarrow]
\mor{B3}{B4}{}[\atright, \dotArrow]
\mor{B4}{B5}{}[\atright, \dotArrow]

\mor{A2}{B4}{}
\mor{A3}{B5}{}
\mor{B2}{A4}{}
\mor{B3}{A5}{}
\enddc
\]
Here replicas $A$ and $B$ both start in state 1.  Replica $A$ issues operations $\msf{mult}(2)$ followed by $\msf{add}(1)$, while $B$ concurrently issues operations $\msf{mult}(3)$ followed by $\msf{add}(4)$.  After receiving each others' messages, the arbitration order dictates that each replica should apply $\msf{add}(1)$ before the concurrent operation $\msf{mult}(3)$, and $\msf{add}(4)$ before $\msf{mult}(2)$.  Additionally, to respect the causal order, replicas should apply $\msf{mult}(2)$ before $\msf{add}(1)$ and $\msf{mult}(3)$ before $\msf{add}(4)$.  However, this creates a loop.

As a second attempt, observe that applying $\msf{add}(m)$ followed by $\msf{mult}(n)$ is the same as applying $\msf{mult}(n)$ followed by $\msf{add}(nm)$, by the distributive property.  Thus in the face of concurrent operations $\msf{add}(m)$ and $\msf{mult}(n)$, instead of requiring all replicas to apply $\msf{add}(m)$ before $\msf{mult}(n)$, we can equivalently require replicas that first received $\msf{mult}(n)$ to apply $\msf{add}(nm)$ when they later receive $\msf{add}(m)$.  Replicas that instead received $\msf{add}(m)$ before $\msf{mult}(n)$ apply both operations normally.  This rule is inspired by the semidirect product of groups (see Appendix \ref{sec:algebra}), in which non-commuting operations from two groups can be transposed so long as we transform the first group's operations by the second's.

Using this rule in the above scenario, both replicas end up in state $17$:
\[
\begindc{\commdiag}[50]
\obj(5,7){$A$}
\obj(5,00){$B$}
\obj(10,7)[A1]{1}
\obj(20,7)[A2]{2}
\obj(30,7)[A3]{3}
\obj(40,7)[A4]{9}
\obj(50,7)[A5]{17}
\obj(10,00)[B1]{1}
\obj(20,00)[B2]{3}
\obj(30,00)[B3]{7}
\obj(40,00)[B4]{14}
\obj(50,00)[B5]{17}

\mor{A1}{A2}{$\msf{mult}(2)$}
\mor{A2}{A3}{$\msf{add}(1)$}
\mor{A3}{A4}{$(\msf{mult}(3))$}[\atleft, \dotArrow]
\mor{A4}{A5}{$(\msf{add}(8))$}[\atleft, \dotArrow]
\mor{B1}{B2}{$\msf{mult}(3)$}[\atright, \solidarrow]
\mor{B2}{B3}{$\msf{add}(4)$}[\atright, \solidarrow]
\mor{B3}{B4}{$(\msf{mult}(2))$}[\atright, \dotArrow]
\mor{B4}{B5}{$(\msf{add}(3))$}[\atright, \dotArrow]

\mor{A2}{B4}{}
\mor{A3}{B5}{}
\mor{B2}{A4}{}
\mor{B3}{A5}{}
\enddc
\]
Here $A$ applies $\msf{add}(8)$ in place of $\msf{add}(4)$ because of its concurrent $\msf{mult}(2)$ operation, and $B$ applies $\msf{add}(3)$ in place of $\msf{add}(1)$ because of its concurrent $\msf{mult}(3)$ operation.

Unlike our first attempt, this approach generalizes to all situations via the rule: upon receiving an operation $\msf{add}(m)$, instead of applying it directly, a replica applies the operation $\msf{add}(n_1n_2 \cdots n_km)$, where $\msf{mult}(n_1), \msf{mult}(n_2), \dots, \msf{mult}(n_k)$ are all of the $\msf{mult}$ operations concurrent to $\msf{add}(m)$ that the replica had previously applied.

\subsection{Construction}
\label{sec:construction}
We obtain our semidirect product construction by generalizing the above approach.  Let $\mc{C}_1 = (\Sigma, \sigma^0, \msf{prepare}_1, \msf{effect}_1, \msf{eval})$ and $\mc{C}_2 = (\Sigma, \sigma^0, \msf{prepare}_2, \msf{effect}_2, \msf{eval})$ be two op-based CRDTs sharing the same state space, initial state, and $\msf{eval}$ function, but with disjoint sets of operations and prepared messages. 

In addition to properties (i) and (ii) of an op-based CRDT, we assume that $\mc{C}_1$ and $\mc{C}_2$ satisfy the following property (iib), which is a strengthening of property (ii).  Here by an \textit{author} of a message $m$, we mean any replica $r$ such that $m = \msf{prepare}(o, \sigma, r)$ for some $o$ and $\sigma$, and we say two messages $m_1, m_2$ \textit{can have different authors} if we can write $m_1 = \msf{prepare}(o_1, \sigma_1, r_1)$ and $m_2 = \msf{prepare}(o_2, \sigma_2, r_2)$ with $r_1 \neq r_2$.
\begin{itemize}
    \item[(iib)] For all $\sigma \in \Sigma$ and all messages $m_1, m_2$ that can have different authors, if $m_1 \cdot \sigma \neq \bot$ and $m_2 \cdot \sigma \neq \bot$, then
    \[
    m_1 \cdot (m_2 \cdot \sigma) = m_2 \cdot (m_1 \cdot \sigma) \neq \bot.
    \]
\end{itemize}
That is, instead of requiring $m_1$ and $m_2$ to commute when applied to $\sigma$ only if they all can appear in some execution of Algorithm \ref{alg:crdt_use} with $m_1$ and $m_2$ concurrent, we require commutativity whenever $m_1$ and $m_2$ are both defined on $\sigma$ and can have different authors.  This strengthened version is necessary because in the semidirect product construction, we combine operations from $\mc{C}_1$ and $\mc{C}_2$, potentially allowing states and concurrent messages that are not possible in either CRDT alone.

Let $M_1$ and $M_2$ be the sets of prepared messages for $\mc{C}_1$ and $\mc{C}_2$, respectively.  We wish to construct a CRDT combining the operations of $\mc{C}_1$ and $\mc{C}_2$, with $M_1$ coming before $M_2$ in the arbitration order, so that in the face of concurrent messages from the two CRDTs, those from $\mc{C}_1$ are effectively applied first.

To make this possible, we assume we are given a partial \textit{action} of $M_2$ on $M_1$, i.e., a partial function $\act: M_2 \times M_1 \pra M_1$.  We write $\act$ as in infix, e.g., $m_2 \act m_1$.  This action will be used to transform a message $m_1 \in M_1$ that is received after a concurrent message $m_2 \in M_2$, to give the same result as if $m_1$ had been applied before $m_2$ like it should have been.

\begin{myex}
In the example of the previous section, $\mc{C}_1$ is an integer register with operations $\msf{add}(m)$, $\mc{C}_2$ is an integer register with operations $\msf{mult}(n)$, and $\act$ is given by
\[
\msf{mult}(n) \act \msf{add}(m) = \msf{add}(nm).
\]
\end{myex}

For the semidirect product to work, we assume:
\begin{itemize}
  \item (reordering) For all $\sigma \in \Sigma$ and all messages $m_1 \in M_1$, $m_2 \in M_2$ that can have different authors (i.e., we can write $m_1 = \msf{prepare}_1(o_1, \sigma_1, r_1)$ and $m_2 = \msf{prepare}_2(o_2, \sigma_2, r_2)$ with $r_1 \neq r_2$), if $m_1 \cdot \sigma \neq \bot$ and $m_2 \cdot \sigma \neq \bot$, then $m_2 \act m_1 \neq \bot$ and
  \[
  m_2 \cdot (m_1 \cdot \sigma) = (m_2 \act m_1) \cdot (m_2 \cdot \sigma) \neq \bot.
  \]
  This ensures that in case of concurrent messages $m_1$ and $m_2$ acting on a state $\sigma$, the intended final state $m_2 \cdot (m_1 \cdot \sigma)$ is defined, and it can be computed either by applying $m_1$ followed by $m_2$ or by applying $m_2$ followed by $m_2 \act m_1$.
  \item (action commutes) For all $m_1 \in M_1$ and all $m_2, m_2' \in M_2$, if $m_1, m_2, m_2'$ can have mutually different authors, then
  \[
  m_2 \act (m_2' \act m_1) = m_2' \act (m_2 \act m_1).
  \]
  This ensures that concurrent $M_2$ messages commute in their action on $M_1$ messages via $\act$, just like how they commute when applied to states in $\Sigma$.
  \item (preserves authors) For all $m_1 = \msf{prepare}_1(o_1, \sigma_1, r_1)$ and $m_2 \in M_2$, if $m_2 \act m_1 \neq \bot$, then $m_2 \act m_1 = \msf{prepare}_1( o_1', \sigma_1', r_1)$ for some $o_1', \sigma_1'$.  This ensures that $\act$ preserves message authors, hence messages $m_1' \in M_1$ that are required to commute with $m_1$ by property (iib) of an op-based CRDT are also required to commute with $m_2 \act m_1$.\footnote{In practice, $m_2 \act m_1$ sometimes exists only for the purposes of the semidirect product and is not an output of $\msf{prepare}_1$, in which case we formally define its author to be that of $m_1$, for the purposes of property (iib) and the above assumptions.}
\end{itemize}

\begin{myex}
In the example of the previous section, assumption (reordering) holds by the distributive property:
\begin{align*}
\msf{mult}(n) \cdot (\msf{add}(m) \cdot \sigma) = n(m + \sigma) = nm + n\sigma \\ = \msf{add}(mn) \cdot (\msf{mult}(n) \cdot \sigma).
\end{align*}
Assumption (action commutes) holds by commutativity of multiplication, and (preserves authors) is trivial as any message can have any author.
\end{myex}

\begin{mydef}
Given the assumptions above, the \emph{semidirect product of $\mc{C}_1$ and $\mc{C}_2$ with respect to $\act$} is the op-based CRDT $\mc{C}_1 \rtimes_\act \mc{C}_2 = (\Sigma_\rtimes, (\sigma^0, \mathbf{0}, \emptyset), \msf{prepare}_\rtimes, \msf{effect}_\rtimes, \msf{eval}_\rtimes)$ with components defined in Algorithm \ref{alg:construction}.
\end{mydef}
\begin{algorithm}[ht!]
\begin{algorithmic}[1]
\State $\mc{T}$: set of timestamps in the form of vector clocks, i.e., functions from replica ids to $\N$
\State $\mathbf{0}$: all-0 (initial) timestamp
\State $\Sigma_\rtimes = \Sigma \times \mc{T} \times \mc{P}_{\text{fin}}(M_2 \times \mc{T})$, where $\mc{P}_{\text{fin}}(M_2 \times \mc{T})$ denotes the set of finite subsets of $M_2 \times \mc{T}$
\State \textbf{function} $\msf{prepare}_\rtimes(o, (\sigma, t, H), r)$:
  \Indent
  \State $t' \gets t[r \mapsto t(r) + 1]$
  \If{$o$ is a $\mc{C}_1$ operation} \State \Return $(\msf{prepare}_1(o, \sigma, r), t')$
  \Else \Comment{$o$ is a $\mc{C}_2$ operation}
    \State \Return $(\msf{prepare}_2(o, \sigma, r), t')$
  \EndIf
  \EndIndent
\State \textbf{function} $\msf{effect}_\rtimes((m, t'), (\sigma, t, H))$:
  \Indent
  \If{$m \in M_2$}
    \State \Return $(\msf{effect}_2(m, \sigma), \max(t, t'), H \cup \{(t', m)\})$
  \Else \Comment{$m \in M_1$}
    \State $(l_1, u_1), \dots, (l_k, u_k) \gets $ all $(l, u) \in H$ with $u$ concurrent to $t'$, in (any) causal order \label{alg_line:act}
    \State $m_{\text{act}} \gets l_k \act (l_{k-1} \act (\cdots (l_1 \act m)\cdots))$ \label{alg_line:m_act}
    \State \Return $(\msf{effect}_1(m_{\text{act}}, \sigma), \max(t, t'), H)$
  \EndIf
  \EndIndent
\State \textbf{function} $\msf{eval}_\rtimes(q, (\sigma, t, H))$: \Return $\msf{eval}(q, \sigma)$
\end{algorithmic}
\caption{Components of the semidirect product $\mc{C}_1 \rtimes_\act \mc{C}_2$.  In $\msf{effect}_\rtimes$, if any portion of the output is $\bot$, then we set the whole output to be $\bot$.}
\label{alg:construction}
\end{algorithm}

The semidirect product functions as follows.  A state $(\sigma, t,  H)$ corresponds to an internal state $\sigma$ with current timestamp $t$ and history of $M_2$ messages $H$.  To apply a message $m \in M_2$, $\msf{effect}_\rtimes$ applies $m$ to $\sigma$ and also stores $m$ in $H$ together with its timestamp.  To apply a message $m \in M_1$, $\msf{effect}_\rtimes$ first acts on $m$ by all concurrent messages in $H$, i.e., all concurrent $M_2$ messages that have already been applied to the state, using timestamps to determine which messages are concurrent.  These concurrent messages $(l_1, u_1), \dots, (l_k, u_k)$ act on $m$ in causal order, by which we mean any ordering such that for all $j < j'$, $u_j$ is not causally prior to $u_{j'}$.  The resulting message $m_{\text{act}}$ is then applied to $\sigma$.

\subsection{Correctness}
\begin{mythm}
\label{thm:correctness}
$\mc{C}_1 \rtimes_\act \mc{C}_2$ is an op-based CRDT in the sense of Definition \ref{def:crdt}.
\end{mythm}
\begin{proof}
First, note that $(m, t') \cdot (\sigma, t, H)$ is the same regardless of which causal ordering we choose for $(l_1, u_1), \dots, (l_k, u_k)$ on line \ref{alg_line:act}.  
Indeed, if $(l_1', u_1'), \dots, (l_k', u_k')$ is another causal ordering of these messages, then we can go from one sequence to the other via a sequence of transpositions $$(l_{j_1}, u_{j_1}), (l_{j_2}, u_{j_2}) \mapsto (l_{j_2}, u_{j_2}), (l_{j_1}, u_{j_1})$$ with $u_{j_1}$ concurrent to $u_{j_2}$.  Since $u_{j_1}$, $u_{j_2}$, and $t$ are mutually concurrent, $l_{j_1}$, $l_{j_2}$, and $m$ must have had mutually different authors in the execution leading to this state.  Hence by assumption (action commutes), such a transposition does not change $m_{\text{act}}$.

We now verify the CRDT properties.  Recall that we use $m \cdot \sigma$ as an abbreviation for $\msf{effect}(m, \sigma)$ when the relevant CRDT is clear from context.
\begin{enumerate}[(i)]
  \item Let $(m, t') = \msf{prepare}_\rtimes(o, (\sigma, t, H), r)$.  If $m \in M_2$, then $m \cdot \sigma \neq \bot$ by property (i) of $\mc{C}_2$, so $(m, t') \cdot (\sigma, t, H) \neq \bot$.  If instead $m \in M_1$, then there are no $(l, u) \in H$ with $u$ concurrent to $t'$, since $t'$ comes causally after all prior timestamps.  Hence $m_{\text{act}} = m$, so $m_{\text{act}} \cdot \sigma = m \cdot \sigma \neq \bot$ by property (i) of $\mc{C}_1$.  Thus $(m, t') \cdot (\sigma, t, H) \neq \bot$.
  \item Let $(\sigma, t, H) \in \Sigma_\rtimes$ be a state and $(l, s')$, $(m, t')$ be messages such that $(\sigma, t, H)$ can appear as a replica state and $(l, s')$ and $(m, t')$ can appear as concurrent messages in an execution of Algorithm \ref{alg:crdt_use}, and $(l, s') \cdot (\sigma, t, H) \neq \bot$ and $(m, t') \cdot (\sigma, t, H) \neq \bot$.  We need to show that
  \begin{align*}
  (l, s') \cdot ((m, t') \cdot (\sigma, t, H)) = (m, t') \cdot ((l, s') \cdot (\sigma, t, H)) \neq \bot.
  \end{align*}
\textbf{Case $l, m \in M_1$:} Let $l_{\text{act}}$ and $m_{\text{act}}$ be as in the definitions of $(l, s') \cdot (\sigma, t, H)$ and $(m, t') \cdot (\sigma, t, H)$, respectively.  By assumption (preserves authors) and property (iib) of $\mc{C}_1$, $l_{\text{act}}$ and $m_{\text{act}}$ commute when applied to $\sigma$, and the claim follows.
  
\textbf{Case $l, m \in M_2$:} By property (iib) of $\mc{C}_2$, $l$ and $m$ commute when applied to $\sigma$, and the claim follows.
    
\textbf{Case $l \in M_1, m \in M_2$:} Let $l_{\text{act}}$ and $l_{\text{act}}'$ be as in the definitions of $(l, s') \cdot (\sigma, t, H)$ and $(l, s') \cdot ((m, t') \cdot (\sigma, t, H))$, respectively.  Since $s'$ is concurrent to $t'$, $(m, t')$ appears among the operations used to compute $l_{\text{act}}'$. Furthermore, since Algorithm \ref{alg:crdt_use} delivers messages in causal order, $t'$ must be maximal among timestamps in $H$, so we can put $m$ last in the sequence used to define $l_{\text{act}}'$.  Thus $l_{\text{act}}' = m \act l_{\text{act}}$.  Also, $l_{\text{act}} \cdot \sigma \neq \bot$ and $m \cdot \sigma \neq \bot$, and by assumption (preserves authors), $l_{\text{act}}$ and $m$ can have different authors.  Hence by assumption (reordering), $m \cdot (l_{\text{act}} \cdot \sigma) = (m \act l_{\text{act}}) \cdot (m \cdot \sigma) \neq \bot$.  The claim follows.
\end{enumerate}
\end{proof}

\section{Guide to Using Semidirect Products}
\label{sec:tutorial}
In this section, we go through the steps a CRDT designer should use to construct a semidirect product CRDT.  We do so in the context of an example promised in the introduction: adding higher-order map operations to a map CRDT.  Additional examples appear in Section \ref{sec:examples}.

Let $\mc{C}$ be a CRDT and $K$ be a set of keys.  Numerous works \cite{riak_datatypes, antidote, json} define a map CRDT $\msf{map}(K, \mc{C})$ with keys $K$ and values in $\mc{C}$, which is a replicated version of a dictionary. 
The states of $\msf{map}(K, \mc{C})$ are partial functions $f$ from $K$ to states of $\mc{C}$ such that all but finitely many values of $f$ are $\bot$.  The operations are of the form $\msf{apply}(k, m)$ for $k \in K$ and $m$ a message of $\mc{C}$, with effect
\[
(\msf{apply}(k, m) \cdot f)(k') = \begin{cases} f(k') &\mbox{if $k' \neq k$} \\ m \cdot f(k) &\mbox{if $k' = k$ and $f(k) \neq \bot$} \\ m \cdot \sigma^0 &\mbox{if $k = k'$ and $f(k) = \bot$.} \end{cases}
\]
That is, $\msf{apply}(k, m)$ applies $m$ to the value at $k$, treating $\bot$ as $\sigma^0$.  Typically, CRDT maps also include an operation to remove a key-value pair, but incorporating such an operation into our construction below is difficult, so we leave it as future work.  

For our novel CRDT, we wish to add higher-order map operations $\msf{homap}(m)$ to $\msf{map}(K, \mc{C})$, where $m$ is a message of $\mc{C}$, with sequential semantics
\[
(\msf{homap}(m) \cdot f)(k) = \begin{cases} \bot &\mbox{if $f(k) = \bot$} \\ m \cdot f(k) &\mbox{otherwise.} \end{cases}
\]
That is, $\msf{homap}(m)$ applies $m$ to every non-$\bot$ value in the map.\footnote{Strictly speaking, the allowable choices of $m$ should be restricted so that $m \cdot f(k) \neq \bot$ for all $k$, and so that $m$ commutes with messages $m'$ that could appear in a concurrent $\msf{homap}(m')$ or $\msf{apply}(k, m')$ operation.}

\begin{myex}
As a potential use case, recall the example from the introduction of a Slack-like application containing multiple ``channels'' that users can join.  We can store the map from channel names to the set of users in each channel as a set-valued map CRDT.  When a new user $A$ is added to the application, we can use $\msf{homap}(\msf{add}(A))$ to add $A$ to every channel, without needing to send separate messages for each channel.  It is trivial to modify the operation to filter by keys, so that, e.g., $A$ is only added to channels whose names start with ``public''.
\end{myex}

As our first step in constructing a semidirect product implementing both $\msf{apply}$ and $\msf{homap}$ operations, we need to partition the operations into two sets so that we can easily construct CRDTs for each set.  Here the obvious split is into $\msf{apply}$ and $\msf{homap}$ operations.  Thus our first component CRDT is $\msf{map}(K, \mc{C})$, which implements the $\msf{apply}$ operations.  For our second component CRDT, we let $\msf{homap}(K, \mc{C})$ have the same state space as $\msf{map}(K, \mc{C})$ but with $\msf{homap}$ operations only.  Since $\mc{C}$ is a CRDT, concurrent $\msf{homap}(m)$ operations commute with each other, so $\msf{homap}(K, \mc{C})$ is indeed a CRDT.  Also, assuming $\mc{C}$ satisfies the strengthened property (iib) of an op-based CRDT, so do $\msf{map}(K, \mc{C})$ and $\msf{homap}(K, \mc{C})$.

It is worth noting that we cannot add $\msf{homap}(m)$ operations to $\msf{map}(K, \mc{C})$ directly, since they do not always commute with concurrent operations $\msf{apply}(k, m')$.  Indeed, suppose $\msf{apply}(k, m')$ initializes the value at $k$, i.e., changes it from $\bot$ to non-$\bot$.  This does not commute with $\msf{homap}(m)$: if $f(k) = \bot$, then
\[
\left(\msf{homap}(m) \cdot (\msf{apply}(k, m') \cdot f)\right)(k) = m \cdot (m' \cdot \sigma^0),
\]
while
\[
\left(\msf{apply}(k, m') \cdot (\msf{homap}(m) \cdot f)\right)(k) = m' \cdot \sigma^0
\]
because $\msf{homap}(m)$ only applies to non-$\bot$ values.

Hence we are led to consider a semidirect product of $\msf{map}(K, \mc{C})$ and $\msf{homap}(K, \mc{C})$.  To do so, we must first choose an arbitration order between the two CRDTs.  That is, we must choose which operation should be applied first in the case of concurrent $\msf{apply}(k, m')$ and $\msf{homap}(m)$ operations.  Either choice would be reasonable, but we find it more interesting to put $\msf{homap}(m)$ operations last in the arbitration order.  This means that they also apply to keys that are initialized concurrently.  

\begin{myex}
In the example above, this semantics ensures that $A$ will also be added to channels that are created concurrently to $A$'s addition, so that they are not left out of any eligible channels.
\end{myex}

Next, we need to choose a partial action $\act$ of $\msf{homap}(m)$ messages on $\msf{map}(K, \mc{C})$ messages satisfying assumption (reordering), i.e., for all map states $f$,
\begin{equation}
\label{reordering_ex}
\msf{homap}(m) \cdot (\msf{map}(k, m') \cdot f) = (\msf{homap}(m) \act \msf{map}(k, m')) \cdot (\msf{homap}(m) \cdot f).
\end{equation}
We focus on this assumption first because it is the most difficult one to satisfy, while assumptions (action commutes) and (preserves authors) are largely technical.  In general, to satisfy this assumption, we often need to modify our component CRDTs $\mc{C}_1$ and $\mc{C}_2$:
\begin{itemize}
    \item Sometimes for messages $m_1$ of $\mc{C}_1$ and $m_2$ of $\mc{C}_2$, there is an obvious partial function $f$ from states to states such that $m_2 \cdot (m_1 \cdot \sigma) = f \cdot (m_2 \cdot \sigma)$ for all states $\sigma$, but $f$ does not correspond to the effect of any message of $\mc{C}_1$.  This is easily resolved by expanding the message set of $\mc{C}_1$ to include $f$.  These ``formal'' messages need not have a corresponding externally visible operation.  Note, however, that we must then ensure that assumption (reordering) also holds for these new messages of $\mc{C}_1$.
    \item Sometimes for messages $m_1$ of $\mc{C}_1$ and $m_2$ of $\mc{C}_2$, it is impossible to satisfy assumption (reordering) because there are states $\sigma \neq \sigma'$ such that $m_2 \cdot (m_1 \cdot \sigma) \neq m_2 \cdot (m_1 \cdot \sigma')$ but $m_2 \cdot \sigma = m_2 \cdot \sigma'$.  This can often be solved by adding extra metadata to states that makes $m_2 \cdot \sigma$ different from $m_2 \cdot \sigma'$.
\end{itemize}

For $\msf{homap}$, observe that if $f(k) = \bot$, then
\[
\left(\msf{homap}(m) \cdot (\msf{apply}(k, m') \cdot f)\right)(k) = m \cdot (m' \cdot \sigma^0) = m' \cdot (m \cdot \sigma^0)
\]
since $m$ and $m'$ commute, while $(\msf{homap}(m) \cdot f)(k) = \bot$ because $\msf{homap}(m)$ only affects initialized values.   Meanwhile, if $f(k) \neq \bot$, then
\[
\left(\msf{homap}(m) \cdot (\msf{apply}(k, m') \cdot f)\right)(k) = m \cdot (m' \cdot f(k)) = m' \cdot (m \cdot f(k)),
\]
while $(\msf{homap}(m) \cdot f)(k) = m \cdot f(k)$.  Thus to satisfy (\ref{reordering_ex}), $\msf{homap}(m) \act \msf{apply}(k, m')$ must act as
\[
((\msf{homap}(m) \act \msf{apply}(k, m')) \cdot g)(k) = \begin{cases} m' \cdot g(k) &\mbox{if $g(k) \neq \bot$} \\ m' \cdot (m \cdot g(k)) &\mbox{if $g(k) = \bot$.} \end{cases}
\]
No message of $\msf{map}(K, \mc{C})$ acts in this way, but as discussed in the first bullet above, we can easily add one that does.  Defining $\act$ on those messages leads us to need more new messages, etc.  Eventually we are led to define the messages of $\msf{map}(K, \mc{C})$ to be of the form $\msf{apply}(k, m)(L)$, where $L$ is a finite sequence of $\mc{C}$ messages.  The effect of such a message on a state $f$ is
\begin{align*}
(\msf{apply}(k, m)(L) \cdot f)(k') = \begin{cases} f(k') &\mbox{if $k' \neq k$} \\ m \cdot f(k) &\mbox{if $k' = k$ and $f(k) \neq \bot$} \\ m \cdot (L \cdot \sigma^0) &\mbox{if $k' = k$ and $f(k) = \bot$,} \end{cases}
\end{align*}
where $L \cdot \sigma^0$ denotes the result of applying all messages in $L$ to $\sigma^0$ in order.  In other words, $\msf{apply}(k, m)(L)$ acts the same as $\msf{apply}(k, m)$, except it also applies $L$ to an uninitialized value.  The original map operations $\msf{apply}(k, m)$ corresponds to the messages $\msf{apply}(k, m)(\emptyset)$.

Now we can define
\[
\msf{homap}(m) \act \msf{apply}(k, m')(L) = \msf{apply}(k, m')(L \cup \{m\}),
\]
where $m$ is appended to the end of $L$.  This trivially satisfies assumptions (action commutes) and (preserves authors), and it satisfies assumption (reordering) as well: both
\[
\left(\msf{homap}(m) \cdot (\msf{apply}(k, m')(L) \cdot f)\right)(k)
\]
and
\[
\left(\msf{apply}(k, m')(L \cup \{m\}) \cdot (\msf{homap}(m) \cdot f)\right)(k)
\]
equal
\[
\begin{cases} m \cdot (m' \cdot f(k)) &\mbox{if $f(k) \neq \bot$} \\ m \cdot (m' \cdot (L \cdot \sigma^0)) &\mbox{if $f(k) = \bot$.} \end{cases}
\]
Thus we get a semidirect product CRDT $\msf{map}(K, \mc{C}) \rtimes_\act \msf{homap}(K, \mc{C})$ with both $\msf{apply}(k, m')$ and $\msf{homap}(m)$ operations, implementing the semantics described above.  In this CRDT, a $\msf{homap}(m)$ operation modifies concurrent $\msf{apply}(k, m')$ operations to ensure that if they initialize the value at $k$, they also apply $m$ to that value.

\section{Examples}
\label{sec:examples}
We now give numerous examples of semidirect products, both constructing novel CRDTs and reproducing the semantics of existing CRDTs.  These examples demonstrate the semidirect product's ability to compose and decompose CRDTs.  In several of the examples, both components are commutative data types (i.e., all messages commute naturally without conflict resolution), in which case the semidirect product handles all conflicts between non-commuting concurrent messages.


\subsection{Novel CRDTs}
\paragraph{Sequence with Reverse Operation}
We can use semidirect products to add a $\msf{reverse}$ operation to a sequence CRDT.  A sequence CRDT is a CRDT version of a totally ordered sequence, such as a text string, as would appear in a collaborative text editing application.

Let $\mc{C}$ be the continuous sequence CRDT defined by Shapiro et~al.\ \cite[\S 3.5.2]{crdt_survey_2011}, with $\R$ as the continuum set of identifiers.  In this CRDT, the state consists of a set $S$ of sequence elements (e.g., characters) tagged with unique identifiers of the form $(x, r)$ for $x \in \R$ and $r$ a replica id.  The identifiers, and their corresponding elements, are totally ordered by their real number component, with ties broken using an arbitrary total order on replica ids.  Calling $\msf{prepare}$ on an operation to insert an element $e$ to the right of a given element $e'$ results in a message $\msf{add}(e, (x, r))$, where $r$ is the generating replica's id and $x$ is a real number halfway between the identifiers for $e'$ and the next element to the right of $e'$.  The effect of $\msf{add}(e, (x, r))$ is to add $e$ to the state with identifier $(x, r)$.  We also have operations $\msf{remove}(x, r)$, which remove the element with identifier $(x, r)$.

For our novel CRDT, we wish to add an operation $\msf{reverse}$ acting as
\[
\msf{reverse} \cdot S = \{(e, (-x, -r)) \mid (e, (x, r)) \in S\},
\]
where ``negative'' replica ids are ordered oppositely to ordinary replica ids.
Let $\mc{C}_{\text{rev}}$ be the commutative data type with the same state space as $\mc{C}$ and with $\msf{reverse}$ as its single operation and message.  We choose $\msf{reverse}$ to come after messages of $\mc{C}$ in the arbitration order, so that concurrent insertions are also reversed.  Next, we must find an action $\act$ of $\msf{reverse}$ on messages of $\mc{C}$ satisfying assumption (reordering).  This is easily done by setting
\begin{align*}
&\msf{reverse} \act \msf{add}(e, (x, r)) = \msf{add}(e, (-x, -r)) \\
&\msf{reverse} \act \msf{remove}(x, r) = \msf{remove}(-x, -r).
\end{align*}
Assumptions (action commutes) and (preserves authors) hold as well.
Thus the semidirect product CRDT $\mc{C} \rtimes_\act \mc{C}_{\text{rev}}$ implements the operations of $\mc{C}$ together with $\msf{reverse}$.

\begin{myrmk}
It is also possible to use the opposite arbitration order, in which $\msf{reverse}$ operations come first and hence do not affect concurrent operations, using a semidirect product in which a message $m$ of $\mc{C}$ acts on a $\msf{reverse}$ operation to add an exception for $m$.  We omit the details.
\end{myrmk}

\paragraph{Sequence with Range Remove Operation}
Let $\mc{C}$ be the continuous sequence CRDT described above.  We can also use the semidirect product to add a range remove operation $\msf{rremove}((x, r), (x', r'))$ to $\mc{C}$, which removes all elements with identifiers $(x, r) \le (y, s) \le (x', r')$.  This could be useful as an optimization in a collaborative text editor when a user highlights and deletes a block of text, in place of sending separate $\msf{remove}$ messages for each character.

Let $\mc{C}_{\text{rremove}}$ be the commutative data type with the same state space as $\mc{C}$ and with operations $\msf{rremove}((x, r), (x', r'))$, acting as described above.  We choose $\msf{rremove}$ operations to come after messages of $\mc{C}$ in the arbitration order, so that they affect concurrent additions.  We define $\act$ by
\begin{align*}
&\msf{rremove}((x, r), (x', r')) \act \msf{add}(e, (y, s)) \\
&\qquad = \begin{cases} \msf{id} &\mbox{if $(x, r) \le (y, s) \le (x', s')$} \\ \msf{add}(e, (y, s)) &\mbox{otherwise} \end{cases}
\\
&\msf{rremove}((x, r), (x', r')) \act \msf{remove}(y, s) = \msf{remove}(y, s),
\end{align*}
where $\msf{id}$ is a new message we add to $\mc{C}$ which acts as the identity.  Then the semidirect product assumptions are easily checked, so we get a semidirect product CRDT $\mc{C} \rtimes_\act \mc{C}_{\text{rremove}}$ implementing the operations of $\mc{C}$ together with $\msf{rremove}$.

As with $\msf{reverse}$, it is also possible to implement the opposite arbitration order, in which $\msf{rremove}$ operations only remove elements that were added causally before the $\msf{rremove}$.

\paragraph{Semirings}
The example of Section \ref{sec:example} generalizes to the case when we replace $(\Z, +, \times)$ by any commutative semiring:
\begin{mydef}
A \textit{commutative semiring} \cite{semiring} is a tuple $(S, \oplus, \otimes)$ consisting of a set of states $S$ and binary operations $\oplus, \otimes: S \times S \ra S$, such that $\oplus$ and $\otimes$ are associative and commutative, and $\otimes$ distributes over $\oplus$, i.e., for all $s, t, u \in S$, $s \otimes (t \oplus u) = (s \otimes t) \oplus (s \otimes u)$.
\end{mydef}
Examples include $(\Z, +, \times)$, $(\N, \min, +)$, $(\N, \max, \min)$, and $(\N, \min, \max)$.  We will define a semidirect product CRDT implementing the operations of any commutative semiring, from which we immediately get CRDTs implementing all of these examples.

Given a commutative semiring $(S, \oplus, \otimes)$, let $\mc{C}_1$ be the commutative data type with state space $S$ and operations $\msf{add}(s)$, $s \in S$, acting as $\msf{add}(s) \cdot \sigma = s \oplus \sigma$.  Similarly let $\mc{C}_2$ be the commutative data type with state space $S$ and operations $\msf{mult}(t)$, $t \in S$, acting as $\msf{mult}(t) \cdot \sigma = t \otimes \sigma$.  Following the example of Section \ref{sec:example}, we can define $\act$ by
\[
\msf{mult}(t) \act \msf{add}(s) = \msf{add}(t \otimes s),
\]
and then we get a semidirect product CRDT $\mc{C}_1 \rtimes_\act \mc{C}_2$.  Here assumption (reordering) holds by the distributive property:
\begin{align*}
\msf{mult}(t) \cdot (\msf{add}(s) \cdot \sigma) = t \otimes (s \oplus \sigma) = (t \otimes s) \oplus (t \otimes \sigma) \\ = \msf{add}(t \otimes s) \cdot (\msf{mult}(t) \cdot \sigma).
\end{align*}
The semidirect product CRDT implements both $\msf{add}(s)$ and $\msf{mult}(t)$ operations on $S$, with $\msf{mult}(t)$ operations affecting the arguments to concurrent $\msf{add}(s)$ operations.

\begin{myrmk}
By iteratively applying a semidirect product construction similar to the semiring construction, it appears to be possible to construct a CRDT supporting operations $\msf{mult}$, $\msf{add}$, $\min$, and $\max$ on a natural number register, with arbitration order $\msf{mult} > \msf{add} > \min > \max$.  We leave the details to future work.
\end{myrmk}


\subsection{Existing CRDTs}

\paragraph{Boolean Flags}
The \textit{enable-wins flag} \cite{riak_datatypes} is a simple CRDT with state space $\Sigma = \{\msf{enabled}, \msf{disabled}\}$, initial state $\msf{disabled}$, and operations $\msf{enable}$ and $\msf{disable}$, with sequential semantics
\begin{align*}
&\msf{enable} \cdot \sigma = \msf{enabled}, &\msf{disable} \cdot \sigma = \msf{disabled}
\end{align*}
for all $\sigma \in \Sigma$.  In case of concurrent $\msf{enable}$ and $\msf{disable}$ operations, the $\msf{enable}$ wins, so that the state is $\msf{enabled}$.  More precisely, letting $\prec$ be the causal order on messages, a replica's state is $\msf{enabled}$ if it has received any messages $m = \msf{enable}$ such that it has not received any messages $m' = \msf{disable}$ with $m \prec m'$; otherwise the state is $\msf{disabled}$.

As our first step in constructing a semidirect product with the same semantics as the enable-wins flag, we need to partition the operations into two sets so that we can easily construct CRDTs for each set.  Since there are only two operations, we take the two sets to be $\{\msf{enable}\}$ and $\{\msf{disable}\}$.  Both of these singleton sets of operations, acting on the original state space $\Sigma$, define commutative data types (i.e., all messages commute naturally without conflict resolution), since trivially $\msf{enable}$ commutes with itself and likewise for $\msf{disable}$.  Hence they also define op-based CRDTs satisfying the strengthened property (iib).

Next, we must choose an arbitration order between the two CRDTs.  That is, we must choose which operation should be applied first in the case of concurrent $\msf{enable}$ and $\msf{disable}$ operations.  The enable-wins semantics corresponds to $\msf{disable}$ going first, so that the subsequent $\msf{enable}$ wins.  Thus we take $\mc{C}_1$ to be the CRDT with operation set $\{\msf{disable}\}$, while $\mc{C}_2$ is the CRDT with operation set $\{\msf{enable}\}$.

Next, we need to find a partial action $\act$ of $\mc{C}_2$ messages on $\mc{C}_1$ messages satisfying assumption (reordering).  As a first attempt, we add an identity message $\msf{id}$ to the message set of $\mc{C}_1$, acting as $\msf{id} \cdot \sigma = \sigma$, and define
\begin{align*}
    &\msf{enable} \act \msf{disable} = \msf{id}, &\msf{enable} \act \msf{id} = \msf{id}.
\end{align*}
This satisfies assumption (reordering) because $$\msf{enable} \circ \msf{disable} = \msf{enable} = \msf{id} \circ \msf{enable}.$$  It also trivially satisfies assumptions (action commutes) and (preserves authors), where we formally allow any replica to be an author of $\msf{id}$.  Hence the semidirect product $\mc{C}_1 \rtimes_{\act} \mc{C}_2$ is defined.  However, it does not quite implement the enable-wins semantics in the following scenario (using $\msf{d}$ and $\msf{e}$ as obvious abbreviations):
\[
\begindc{\commdiag}[50]
\obj(5,7){$A$}
\obj(5,00){$B$}
\obj(10,7)[A1]{$\msf{d}$}
\obj(20,7)[A2]{$\msf{e}$}
\obj(30,7)[A3]{$\msf{d}$}
\obj(40,7)[A4]{$\msf{e}$}
\obj(50,7)[A5]{$\msf{e}$}
\obj(10,00)[B1]{$\msf{d}$}
\obj(20,00)[B2]{$\msf{e}$}
\obj(30,00)[B3]{$\msf{d}$}
\obj(40,00)[B4]{$\msf{e}$}
\obj(50,00)[B5]{$\msf{e}$}

\mor{A1}{A2}{$\msf{e}$}
\mor{A2}{A3}{$\msf{d}$}
\mor{A3}{A4}{$(\msf{e})$}[\atleft, \dotArrow]
\mor{A4}{A5}{$(\msf{id})$}[\atleft, \dotArrow]
\mor{B1}{B2}{$\msf{e}$}[\atright, \solidarrow]
\mor{B2}{B3}{$\msf{d}$}[\atright, \solidarrow]
\mor{B3}{B4}{$(\msf{e})$}[\atright, \dotArrow]
\mor{B4}{B5}{$(\msf{id})$}[\atright, \dotArrow]

\mor{A2}{B4}{}
\mor{A3}{B5}{}
\mor{B2}{A4}{}
\mor{B3}{A5}{}
\enddc
\]
Here replicas $A$ and $B$ both start in state $\msf{disabled}$, then each concurrently issue operations $\msf{enable}$ followed by $\msf{disable}$.  After receiving each others' messages, the intended result is $\msf{disabled}$, since both $\msf{enable}$ operations have been overwritten by causally greater $\msf{disable}$ operations.  However, the $\msf{disable}$ messages instead both get transformed to $\msf{id}$ by the concurrent $\msf{enable}$ messages, resulting in state $\msf{enabled}$.

To resolve this, observe that in the enable-wins flag, the effect of a $\msf{disable}$ message is not always to set the state to $\msf{disabled}$.  Instead, it cancels the effect of any causally lesser $\msf{enable}$ messages, resulting in state $\msf{disabled}$ only if no $\msf{enable}$ messages remain.

This motivates us to replace $\Sigma$ with the state space $\Sigma'$ whose states are sets $S$ of $\msf{enable}$ messages.  The externally visible value of a state $S$ is $\msf{enabled}$ if $S \neq \emptyset$ and $\msf{disabled}$ if $S = \emptyset$.  In place of $\mc{C}_1$, we define the CRDT $\mc{C}_1'$ which has a single operation and message $\msf{disable}$ acting as $\msf{disable} \cdot S = \emptyset$.  In place of $\mc{C}_2$, we define the CRDT $\mc{C}_2'$ which has a single operation $\msf{enable}$ and message space $M_2'$ containing infinitely many messages of the form $\msf{enable}$\footnote{We need infinitely many of them so that it makes sense to talk about ``the set of $\msf{enable}$ messages received so far'' as a set containing one copy of $\msf{enable}$ for each time a replica issues an $\msf{enable}$ operation.}, acting as $m \cdot S = S \cup \{m\}$.  Note that both $\mc{C}_1'$ and $\mc{C}_2'$ are commutative data types.

We once again choose $\mc{C}_1'$ (the $\msf{disable}$ operations) to come first in the arbitration order.  Thus the semidirect product transformation $\act'$ must satisfy, for $m \in M_2'$ and $s \in \Sigma'$,
\[
m \cdot (\msf{disable} \cdot S) = (m \act' \msf{disable}) \cdot (m \cdot S),
\]
i.e., $\{m\} = (m \act \msf{disable}) \cdot (S \cup \{m\})$.  We see that $m \act \msf{disable}$ should be a message that intersects the state with $\{m\}$.  No $\mc{C}_1'$ message does this, but we can easily add one that does.  This leads us to need more new messages, etc.  Eventually we are led to define the messages of $\mc{C}_1'$ to be $M_1' = \{\msf{disable}(S') \mid S' \in \Sigma'\}$, with $\msf{disable}(S')$ acting as
\[
\msf{disable}(S') \cdot S = S \cap S'.
\]
The original operation $\msf{disable}$ is prepared as the message $\msf{disable}(\emptyset)$.  Then we can satisfy assumption (reordering) by setting
\[
m \act \msf{disable}(S') = \msf{disable}(S' \cup \{m\}).
\]
Assumptions (action commutes) and (preserves authors) are trivially satisfied as well, where we again formally allow any replica to be an author of $\msf{disable}(S')$.  Thus we get a semidirect product CRDT
\[
\mc{C}_1' \rtimes_{\act'} \mc{C}_2'.
\]
In this CRDT, the effect of an $\msf{enable}$ operation is to add itself to the internal state, making the externally visible value $\msf{enabled}$.  When a replica receives a message $\msf{disable}(\emptyset)$ corresponding to a $\msf{disable}$ operation, it first acts by all concurrent $\msf{enable}$ messages.  This leads to the message $\msf{disable}(S')$, where $S'$ is the set of all $\msf{enable}$ messages concurrent to the $\msf{disable}$ operation.  That message is then applied to the state, thus removing any $\msf{enable}$ messages that were causally lesser than the $\msf{disable}$ operation.  If this removes all $\msf{enable}$ messages, the internal state becomes $\emptyset$ and the externally visible value becomes $\msf{disabled}$.  Thus we have indeed implemented the enable-wins semantics as a semidirect product of commutative data types.

\begin{myrmk}
The internal state $S$ duplicates the role of the history set $H$ in the semidirect product $\mc{C}_1' \rtimes_{\act'} \mc{C}_2'$, except that it is trimmed to only contain relevant (not yet disabled) messages.  Thus we can optimize the construction, without changing its semantics, by allowing $S$ to double as $H$.
\end{myrmk}

We can likewise decompose the \textit{disable-wins flag}, which is identical but with the roles of $\msf{enable}$ and $\msf{disable}$ switched.

\paragraph{Sets}
Using a similar decomposition to the enable-wins flag, we can decompose two set CRDTs, the \textit{add-wins set} and \textit{remove-wins set}.  These have a set as their externally visible state and operations $\msf{add}(a)$ and $\msf{remove}(a)$ for $a$ in a universe of set elements, with the obvious sequential semantics.  In the add-wins set, $\msf{add}(a)$ operations win over concurrent $\msf{remove}(a)$ operations, like in the enable-wins flag.  More precisely, letting $\prec$ be the causal order on messages, the value of an add-wins set after receiving messages $H$ is
\[
\{a \mid \exists (m = \msf{add}(a)) \in H.\ \forall (m' = \msf{remove}(a)) \in H.\ m \nprec m'\}.
\]
In other words, the add-wins set functions like an enable-wins flag for each set element.

This leads us to decompose the add-wins set as the following semidirect product of commutative data types.  Let $\Sigma$ be the state space whose states are sets $S$ of $\msf{add}(a)$ messages, for $a$ in the universe of set elements.  The externally visible value of a state $S \in \Sigma$ is $\{a \mid \exists (m = \msf{add}(a)) \in S\}$.  Let $\mc{C}_2$ be the commutative data type with state space $\Sigma$, operations $\msf{add}(a)$ for $a$ in the universe, and infinitely many messages of the form $\msf{add}(a)$ for each $a$, acting as $m \cdot \Sigma = \Sigma \cup \{m\}$.  Let $\mc{C}_1$ be the commutative data type with state space $\Sigma$, messages $\msf{remove}(a)(S')$ for $S'$ a set of $\msf{add}(a)$ messages, acting as
\[
\msf{remove}(a)(S') \cdot S = \{m \in S \mid m = \msf{add}(b), b \neq a\} \cup (S \cap S'),
\]
and operations $\msf{remove}(a)$ with corresponding prepared message $\msf{remove}(a)(\emptyset)$.  Define the semidirect product action $\act$ by
\begin{align*}
(m = \msf{add}(a)) \act \msf{remove}(b)(S') = \begin{cases} \msf{remove}(b)(S') &\mbox{if $b \neq a$} \\ \msf{remove}(b)(S' \cup \{m\}) &\mbox{if $b = a$}. \end{cases}
\end{align*}

Then it is easy to check assumptions (reordering), (action commutes), and (preserves authors), so we get a semidirect product CRDT $\mc{C}_1 \rtimes_\act \mc{C}_2$.  This CRDT implements the add-wins semantics, as follows.  The effect of an $\msf{add}(a)$ operation is to add itself to the internal state, making the externally visible value contain $a$.  When a replica receives a message $\msf{remove}(a)(\emptyset)$ corresponding to a $\msf{remove}(a)$ operation, it first acts by all concurrent $\msf{add}$ messages.  This leads to the message $\msf{remove}(a)(S')$, where $S'$ is the set of all $\msf{add}(a)$ messages concurrent to the $\msf{remove}(a)$ operation.  That message is then applied to the state, thus removing any $\msf{add}(a)$ messages that were causally lesser than the $\msf{remove}(a)$ operation.  If this removes all $\msf{add}(a)$ messages, the externally visible state no longer includes $a$.

The remove-wins set, in which the roles of $\msf{add}$ and $\msf{remove}$ are switched, can be decomposed analogously.

\paragraph{Reset-Wins Resettable CRDTs}
Let $\mc{C}$ be any CRDT, with state space $\Sigma$, message set $M$, and initial state $\sigma^0$.  We wish to add a $\msf{reset}$ operation to $\mc{C}$ which restores its state to $\sigma^0$.  This is necessary for constructing map CRDTs with values in $\mc{C}$ \cite{riak_datatypes, antidote}: removing a key triggers a reset on the corresponding value.

One possible concurrency semantics is \textit{reset-wins}: any operation concurrent to a $\msf{reset}$ operation is ignored.  When used in a map CRDT, this leads to the \textit{remove-wins semantics}, in which operations on a value have no effect if its key is concurrently removed \cite[\S 2.1.5]{crdt_overview_preguica}.

To construct a reset-wins CRDT from $\mc{C}$, we divide the operations into two sets, one containing the original operations on $\mc{C}$ and one containing the single operation $\msf{reset}$.  For the reset-wins semantics, we want $\msf{reset}$ operations to come last in the arbitration order, so that they overwrite concurrent $\mc{C}$ operations.  This leads us to take the first component of the semidirect product to be $\mc{C}$, and we take the second component to be the CRDT $\mc{C}_{\text{reset}}$ with state space $\Sigma$ and a single operation $\msf{reset}$ acting as $\msf{reset} \cdot \sigma = \sigma^0$.  Note that $\mc{C}_{\text{reset}}$ is a commutative data type.

To define the action $\act$ of $\{\msf{reset}\}$ on $M$, we expand $M$ to contain a message $\msf{id}$ acting as $\msf{id} \cdot \sigma = \sigma$, and we set
\[
\msf{reset} \act m = \msf{id}
\]
for all $m \in M$.  This satisfies assumption (reordering), since for all $m \in M$ and $\sigma \in \Sigma$,
\[
\msf{reset} \cdot (m \cdot \sigma) = \sigma^0 = \msf{id} \cdot (\msf{reset} \cdot \sigma).
\]
Assumptions (action commutes) and (preserves authors) are easily checked, where we formally allow any replica to be an author of $\msf{id}$.  Thus we get a semidirect product CRDT $\mc{C} \rtimes_\act \mc{C}_{\text{reset}}$ implementing the reset-wins semantics.

\paragraph{Observed-Reset Resettable CRDTs}
An alternative concurrency semantics for resets is \textit{observed-reset}: a $\msf{reset}$ only affects causally prior (i.e., observed) operations.  The corresponding \textit{observed-remove} semantics is used in many map CRDTs when a value's key is removed, including for most value types in Riak \cite{riak_datatypes}, Antidote \cite{antidote}, and the JSON CRDT of Kleppmann and Beresford \cite{json}.

The observed-reset semantics corresponds to the arbitration order in which $\msf{reset}$ comes before concurrent operations of $\mc{C}$, so that the concurrent operations are not overwritten by the $\msf{reset}$.  This is the opposite of the reset-wins arbitration order.

For the observed-reset semantics to make sense, we assume that for any execution of $\mc{C}$ following Algorithm \ref{alg:crdt_use}, any subset of the messages appearing in that execution can be applied to $\sigma^0$ in casual order, resulting in a defined state (not $\bot$).  This ensures that if some messages are issued concurrently to a $\msf{reset}$, they can still be applied to the reset state $\sigma^0$.

Our semidirect product construction is similar to the enable-wins flag and add-wins set.  Let $\mc{C}'$ be the same as $\mc{C}$ but with an extra component of the state storing a finite sequence $L'$ of messages from $M$.  When a message is applied to a state $(\sigma, L')$, it is appended to $L'$ in addition to acting on $\sigma$.  Let $\mc{C}_{\text{reset}}'$ be a CRDT with the same state space as $\mc{C}'$ and with messages $\msf{reset}(L)$ for $L$ a finite causally ordered sequence of messages from $M$, acting as
\[
\msf{reset}(L) \cdot (\sigma, L') = ((L \cap L') \cdot \sigma^0, L \cap L')
\]
where $L \cap L'$ denotes $L'$ restricted to messages appearing in $L$, and $(L \cap L') \cdot \sigma^0$ denotes the result of applying the messages in $L \cap L'$ to $\sigma^0$ in order.  By assumption, $(L \cap L') \cdot \sigma^0 \neq \bot$.  Also, $\mc{C}_{\text{reset}}'$ has a single operation $\msf{reset}$ with corresponding prepared message $\msf{reset}(\emptyset)$.  Observe that $\mc{C}_{\text{reset}}'$ is a commutative data type.

Define the semidirect product action $\act$ by
\[
m \act \msf{reset}(L) = \msf{reset}(L \cup \{m\}),
\]
with $m$ appended to the end of the list.  Then the semidirect product $\mc{C}_{\text{reset}}' \rtimes_\act \mc{C}'$ has the operations of $\mc{C}$ plus $\msf{reset}$, and it implements the observed-reset semantics.  Indeed, before a $\msf{reset}$ message is applied to the state, it is modified to ignore all concurrent messages that have already been applied to the state, and any concurrent messages that are later applied to the state are unaffected.

As with the enable-wins flag, we can optimize the construction by allowing the sequence $L'$ appearing in the state to double as the semidirect product's history set.



\section{Generality}

\subsection{Interpretation as Operational Transformation}
\label{sec:ot}
Operational Transformation (OT) is an alternative technique for developing replicated data types that is often viewed as an opposing technique to CRDTs.  OT predates CRDTs and is commonly used in applications, such as Google Docs.  However, general OT is complicated.  As a result, many OT algorithms have turned out to be incorrect \cite{ot_correctness}.  CRDTs were introduced to avoid the complexities and errors of OT \cite{treedoc_crdt, crdt_survey_2011}, by using extra metadata in states as well as prepared messages in place of operations, and by requiring commutativity of concurrent messages instead of transformation properties.

It is thus interesting that, even though the semidirect product is a CRDT construction that reproduces the semantics of many existing CRDTs, it can be viewed as a restricted kind of OT, as we now describe.

For our definition of OT, we use the framework of Ressel, Nitsche-Ruhland, and Gunzenh\"auser \cite{ot_ressel}.  In this framework, to define an OT object with operations $M_1 \cup M_2$\footnote{Unlike op-based CRDTs, OT objects typically do not differentiate between operations and messages, regarding them as the same.}, we must define a transformation function $$tf_1: (M_1 \cup M_2) \times (M_1 \cup M_2) \ra M_1 \cup M_2$$ satisfying Transformation Properties 1 and 2 in \cite{ot_ressel} (copied in Theorem \ref{thm:ot} below).  When a replica receives an operation $m$ from another replica, it transforms $m$ by concurrent operations already in the history according to the adOPTed-algorithm \cite[Figure 8]{ot_ressel}.  It then applies the resulting operation $m'$ to its state and stores $m$.

\begin{mydef}
Let $\mc{C}_1 \rtimes_\act \mc{C}_2$ be a semidirect product of CRDTs.  Let $M_1$ be the messages of $\mc{C}_1$ and $M_2$ those of $\mc{C}_2$.  We define the \textit{semidirect product transformation} on $M_1 \cup M_2$ by
\[
tf_1(m, l) = \begin{cases} l \act m &\mbox{if $m \in M_1$ and $l \in M_2$} \\ m &\mbox{otherwise.} \end{cases}
\]
\end{mydef}
This corresponds to the fact that when an operation $m \in M_1$ is applied after a concurrent operation $l \in M_2$, we apply $l \act m$ instead of $m$.

Our CRDT construction essentially implements the operational transformation object corresponding to $tf_1$, except that we only store messages from $M_2$ in the history.
\begin{mythm}
\label{thm:ot}
The semidirect product transformation $tf_1$ satisfies Transformation Properties 1 and 2 of \cite{ot_ressel}, i.e.,\footnote{\cite{ot_ressel} does not explicitly address the possibility that an operation may be undefined ($\bot$) on a state, or that only certain combinations of operations can be concurrent.  We slightly weaken the Transformation Properties to permit these possibilities.}
\begin{enumerate}[1.]
  \item For all potentially concurrent $l, m \in M_1 \cup M_2$ and $\sigma \in \Sigma$ such that $l \cdot \sigma \neq \bot$ and $m \cdot \sigma \neq \bot$, $$l \cdot (tf_1(m, l) \cdot \sigma) = m \cdot (tf_1(l, m) \cdot \sigma)$$
  \item For all potentially concurrent $k, l, m \in M_1 \cup M_2$, $$tf_1(tf_1(l, k), tf_1(m, k)) = tf_1(tf_1(l, m), tf_1(k, m)).$$
\end{enumerate}
Thus by \cite[Theorem 1]{ot_ressel}, the corresponding OT object is eventually consistent.  Furthermore, $\mc{C}_1 \rtimes_\act \mc{C}_2$ has the same semantics as this OT object.

Conversely, suppose we have an OT object whose operations can be partitioned into disjoint sets $O_1$ and $O_2$, such that the transformation function $tf_1$ satisfies $$\mbox{$tf_1(o, p) = o$ unless $o \in O_1$ and $p \in O_2$.}$$  Then letting $\mc{C}_1$ be the restriction of the OT object to $O_1$ operations (with $O_1$ also as the set of messages), $\mc{C}_2$ be its restriction to $O_2$ operations, and $\act: O_2 \times O_1 \ra O_1$ be given by
\[
o_2 \act o_1 = tf_1(o_1, o_2),
\]
we have that $\mc{C}_1$ and $\mc{C}_2$ are op-based CRDTs, and $\mc{C}_1 \rtimes_\act \mc{C}_2$ is a semidirect product CRDT with the same semantics as the OT object.
\end{mythm}
\begin{proof}
For the first statement, verifying the Transformation Properties is a simple case analysis.  Transformation Property 1 holds when $l, m \in M_1$ because $\mc{C}_1$ is a CRDT, and likewise for $l, m \in M_2$.  The interesting case is when $l \in M_1$ and $m \in M_2$ or vice-versa, in which case it reduces to assumption (reordering).  Transformation Property 2 is trivial in all cases except when $l \in M_1$ and $k,m \in M_2$, in which case it reduces to assumption (action commutes).  (Assumption (preserves authors) only matters in that it ensures that ``potential concurrency'' behaves nicely under transformation.)

For the converse statement, $\mc{C}_1$ and $\mc{C}_2$ are easily CRDTs: if $o, o' \in O_1$, then $tf_1(o, o') = o$ and $tf_1(o', o) = o'$, so $o$ and $o'$ commute by Transformation Property 1, and similarly for $\mc{C}_2$.  The semidirect product assumption (reordering) holds by Transformation Property 1 again, assumption (action commutes) holds by Transformation Property 2, and (preserves authors) holds trivially as any message can have any author.

It remains to see, for both the CRDT-to-OT conversion and the converse conversion, that $\mc{C}_1 \rtimes_\act \mc{C}_2$ has the same semantics as the corresponding OT object.  The adOPTed-algorithm transforms a received operation $o$ by all concurrent operations in the history, except that these operations must themselves be transformed by appropriate operations before they are used as transformers.  While this can be complicated in general, in our case it does not matter: $M_2$ messages (resp.\ $O_2$ operations, for the OT-to-CRDT conversion) are always transformed trivially, and $M_1$ messages (resp.\ $O_1$ operations) never alter the target of their transformation.  Hence the result of the adOPTed transformation applied to $m$ is always $m$ if $m \in M_2$, and it is the $m_{\text{act}}$ appearing on line \ref{alg_line:m_act} of Algorithm \ref{alg:construction} if $m \in M_1$.  Thus the semidirect product CRDT and the OT object end up applying the same message (resp.\ operation) to the state. 
\end{proof}

\begin{myrmk}
Although the semidirect product is a restricted kind of OT, it avoids OT's pitfalls.  First, the restriction on $tf_1$ is severe: while general OT allows $tf_1$ to do anything, we require the result of $tf_1(l, m) \circ m$ to always be equivalent to either $l \circ m$ or $m \circ l$, according to the arbitration order.  This in turn simplifies Transformation Property 2 to assumption (action commutes), which is trivial to verify in all of our examples below, while for general OT algorithms it is often a source of incorrectness.  Second, we allow the semidirect product's components to be full-fledged CRDTs, not just user operations acting on user-visible state.  This allows us to use CRDTs to implement portions of the state that CRDTs are good at, like sequences, while using the semidirect product for conflicts that are best resolved using transformation.  Indeed, it appears impossible to construct a sequence CRDT using semidirect products alone, since there is no clear arbitration order between conflicting insertions.

As a result, the semidirect product is more CRDT-like in character.  Indeed, it reproduces the semantics of several existing CRDTs (see Section \ref{sec:examples}).  Additionally, we avoid the adOPTed-algorithm's need to store a multidimensional history of transformed operations computed using a doubly-recursive algorithm.  Instead, our history is a subset of the messages actually sent by replicas, and these messages transform other messages directly instead of needing to be transformed recursively.
\end{myrmk}

\subsection{Decomposing POLog CRDTs as Semidirect Products}
One existing general model for constructing op-based CRDTs is the POLog (partially ordered log) model of Baquero, Almeida, and Shoker \cite{pure_op_based_crdts_extended}.  In that model, the externally visible state of a CRDT is defined as a function of the log of operations partially ordered by causality.  We used this model implicitly when defining the semantics of the add-wins set in Section \ref{sec:examples}: letting $\prec$ be the causal order on messages, the value of an add-wins set after receiving messages $H$ is
\[
f_{\text{aw-set}}(H, \prec) := \{a \mid \exists (m = \msf{add}(a)) \in H.\ \forall(m' = \msf{remove}(a)) \in H.\ m \nprec m'\}.
\]
The POLog model advocates using $f_{\text{aw-set}}$ directly to implement an add-wins set, by applying it to the current message history each time a user queries the set's state.  This is in contrast to traditional op-based CRDT designs, which typically store a metadata-enhanced version of the original data type's state instead of the full message history.

Given a POLog CRDT, it is interesting to ask whether the CRDT can be decomposed as a semidirect product of simpler CRDTs, by which we mean CRDTs with fewer operations.  Ideally, we would like to repeat this decomposition until we get an iterated semidirect product of commutative data types.  We can then build up an alternate construction of the original CRDT using these semidirect products.  This can clarify the semantics and suggest a more efficient implementation of the original CRDT, by reasoning about conflicts between concurrent messages in a restricted, uniform way instead of allowing the full power of a POLog function.

The following proposition gives a general condition under which a POLog CRDT can be decomposed as the semidirect product of two simpler CRDTs.
\begin{myprop}
Let $f$ be a function defining a POLog CRDT $\mc{C}$ with operation set $O_1 \cup O_2$, i.e., $f$ is a function mapping a partially ordered log of operations in $O_1 \cup O_2$ to an externally visible state.  Suppose that $f(L) = f(L')$ whenever $L$ and $L'$ are partially ordered logs differing only in that, for some message $m_1$ corresponding to an $O_1$ operation and some message $m_2$ corresponding to an $O_2$ operation, $m_1 \prec m_2$ in the partial order of $L$ while $m_1$ is concurrent to $m_2$ in the partial order of $L'$.  Then $\mc{C}$ has the same externally visible semantics as some semidirect product of $\mc{C}_1$ and $\mc{C}_2$, where $\mc{C}_1$ (resp.\ $\mc{C}_2$) is a CRDT with the same externally visible semantics as the restriction of $\mc{C}$ to $O_1$ operations (resp.\ $O_2$ operations).
\end{myprop}
\begin{proof}[Proof Sketch]
To define $\mc{C}_1$ and $\mc{C}_2$, we start with the POLog CRDTs derived from the restriction of $f$ to $O_1$ operations (resp., $O_2$ operations), modified so that their partially ordered logs may each contain both $O_1$ and $O_2$ operations.  We then modify the messages of $\mc{C}_1$ so that in addition to an $O_1$ operation and a timestamp, they contain a set $S$ of $\mc{C}_2$ messages.  We define the semidirect product action by
\[
m_2 \act (m_1, S) = (m_1, S \cup \{m_2\}).
\]
Finally, we define the $\msf{effect}$ of $(m_1, S)$ to be to add $m_1$ to the log with partial order relations: $m_1' \prec m_1$ for all $m_1'$ corresponding to $O_1$ operations with lower timestamps; $m_2 \prec m_1$ for all $m_2$ corresponding to $O_2$ operations that are not in $S$; and any additional relations required by transitivity (so that $\prec$ remains a partial order).  Note that the timestamps on $O_1$ operations are unrelated to those on $O_2$ operations, so that $\act$ is the only way we can reason about concurrent $O_1$ and $O_2$ operations.

After performing an execution on both $\mc{C}$ and $\mc{C}_1 \rtimes_\act \mc{C}_2$, the resulting partially ordered logs $L$ and $L^\rtimes$ differ only in that, for some messages $m_1, m_2 \in L$ corresponding to an $O_1$ operation and an $O_2$ operation, respectively, $m_1 \prec m_2$ in $L$ but $m_1$ is concurrent to $m_2$ in $L^\rtimes$.  Thus by hypothesis, $f(L) = f(L^\rtimes)$.
\end{proof}

Of the non-commutative POLog CRDTs described in \cite{pure_op_based_crdts_extended}, only the multi-value register cannot be decomposed using this proposition, while the enable-wins flag, disable-wins flag, add-wins set, and remove-wins set decompose into semidirect products of commutative data types, as described in Section \ref{sec:examples}.

\section{Optimizations}
\label{sec:optimizations}
As defined above, a state of $\mc{C}_1 \rtimes_\act \mc{C}_2$ includes the set $H$ of all $M_2$ messages that have already been applied.  This set can grow without bound, potentially making the state large and affecting the performance of an implementation.  We now discuss two optimizations that can reduce this state size.

\subsection{Causal Stability}
In a state $(\sigma, t, H)$ of $\mc{C}_1 \rtimes_\act \mc{C}_2$, observe that a pair $(m, t') \in H$ only matters when we apply a message concurrent to it.  Thus once $(m, t')$ becomes \textit{causally stable} \cite[\S 5.2]{pure_op_based_crdts_extended}, meaning that all future inputs $(l, s')$ to $\msf{effect}_\rtimes$ will be causally greater than $(m, t')$, we can discard $(m, t')$ from $H$ without changing the externally visible behavior of $\mc{C}_1 \rtimes_\act \mc{C}_2$.

\subsection{Compressing the History}
In some cases, instead of storing the history of $M_2$ messages in our state as a set $H$, we can store a single $M_2$ message representing the composition of all of these messages.  Specifically, assume:
\begin{itemize}
  \item $M_2$ is closed under composition, in the sense that for all $m_2, m_2' \in M_2$, there exists a message $m_2 \circ m_2' \in M_2$ such that for all $\sigma \in \Sigma$, $(m_2 \circ m_2') \cdot \sigma = m_2 \cdot (m_2' \cdot \sigma)$, and for all $m_1 \in M_1$, $(m_2 \circ m_2') \act m_1 = m_2 \act (m_2' \act m_1)$.
  \item $M_2$ messages commute (not just when they are concurrent), i.e., for all $m_2, m_2' \in M_2$, $m_2 \circ m_2' = m_2' \circ m_2$.
  \item For all $m_2 \in M_2$, the function $m_2' \mapsto m_2 \circ m_2'$ is injective.  We let $m_2^{-1}$ be a formal symbol acting as the corresponding inverse partial function, i.e., $m_2^{-1} \circ (m_2 \circ m_2') := m_2'$.
\end{itemize}
\begin{myex}
The example of Section \ref{sec:example} satisfies these assumptions if we exclude $\msf{mult}(0)$, with $\msf{mult}(n) \circ \msf{mult}(n') = \msf{mult}(nn')$ and with $\msf{mult}(n)^{-1} \circ \msf{mult}(m) = \msf{mult}(m/n)$ when $n$ divides $m$.
\end{myex}

To accommodate the initial state, we formally add the identity function $\msf{id}$ to $M_2$.
\begin{mydef}
The \textit{compressed semidirect product of $\mc{C}_1$ and $\mc{C}_2$} is the op-based CRDT $\mc{C}_1 \rtimes_\act^{\text{comp}} \mc{C}_2 = (\Sigma \times M_2, (\sigma^0, \msf{id}), \msf{prepare}_{\text{comp}}, \msf{effect}_{\text{comp}}, \msf{eval}_{\text{comp}})$ with components defined in Algorithm \ref{alg:compressed}.
\end{mydef}
\begin{algorithm}[ht!]
\begin{algorithmic}[1]
\State \textbf{function} $\msf{prepare}_{\text{comp}}(o, (\sigma, h), r)$:
  \Indent
  \If{$o$ is a $\mc{C}_1$ operation} \State \Return $(\msf{prepare}_1(o, \sigma, r), h)$
  \Else \Comment{$o$ is a $\mc{C}_2$ operation}
    \State \Return $(\msf{prepare}_2(o, \sigma, r), -)$
  \EndIf
  \EndIndent
\State \textbf{function} $\msf{effect}_{\text{comp}}((m, h'), (\sigma, h))$:
  \Indent
  \If{$m \in M_2$}
    \State \Return $(\msf{effect}_2(m, \sigma), h' \circ h)$
  \Else \Comment{$m \in M_1$}
    \State $m_{\text{act}} \gets ((h')^{-1} \circ h) \act m$
    \State \Return $(\msf{effect}_1(m_{\text{act}}, \sigma), h)$
  \EndIf
  \EndIndent
\State \textbf{function} $\msf{eval}_{\text{comp}}(q, (\sigma, h))$: \Return $\msf{eval}(q, \sigma)$
\end{algorithmic}
\caption{Components of the compressed semidirect product $\mc{C}_1 \rtimes_\act^{\text{comp}} \mc{C}_2$.  In $\msf{effect}_{\text{comp}}$, if any portion of the output is $\bot$, then we set the whole output to be $\bot$.}
\label{alg:compressed}
\end{algorithm}

One can show that $\mc{C}_1 \rtimes_\act^{\text{comp}} \mc{C}_2$ satisfies the CRDT properties using a proof similar to Theorem \ref{thm:correctness}.  Informally, instead of storing the history $H$, we store the composition $h$ of all messages in $H$.  When applying an $M_1$ message $(m, h')$ to a state $(\sigma, h)$, $h'$ is the composition of all $M_2$ messages causally prior to $(m, h')$, while $h$ is the composition of all $M_2$ messages applied to the state.  Thus $(h')^{-1} \circ h$ is the composition of all $M_2$ messages applied to the state that are concurrent to $(m, h')$ instead of causally prior to it, so that $m_{\text{act}}$ is the same as it would be in $\mc{C}_1 \rtimes_\act \mc{C}_2$.

\begin{myex}
We can use the compressed semidirect product for the following examples above:
\begin{itemize}
    \item The sequence with reverse operation, since $\msf{reverse}$ is closed under composition (after adding an identity operation $\msf{id}$), commutative, and invertible.  The resulting CRDT is essentially equivalent to treating the possibly-reversed sequence as a view of an ordinary sequence CRDT, with $\msf{reverse}$ operations toggling the view, and with user inputs reversed whenever they are performed on a reversed view.
    \item The semiring CRDT, whenever the semiring $(S, \oplus, \otimes)$ is such that for all $s \in S$, the function $t \mapsto s \otimes t$ is injective.  This includes the example of Section \ref{sec:example} if we exclude $\msf{mult}(0)$ operations, which corresponds to the semiring $(\Z, +, \times)$.
    
    It also includes the semiring $(\N, \min, +)$.  This can be used as a form of resettable counter, with $\min(0)$ behaving as a reset-to-0 operation.  It is practically interesting because we can implement this CRDT with constant-sized state, in contrast to existing resettable counter designs.  However, it does not have either of the typical resettable CRDT semantics described in Section \ref{sec:examples} (reset-wins or observed-reset).  It is similar to observed-reset, but with the following anomaly:
\[
\begindc{\commdiag}[50]
\obj(5,7){$A$}
\obj(5,00){$B$}
\obj(10,7)[A1]{0}
\obj(20,7)[A2]{1}
\obj(30,7)[A3]{0}
\obj(40,7)[A4]{1}
\obj(50,7)[A5]{1}
\obj(10,00)[B1]{0}
\obj(20,00)[B2]{1}
\obj(30,00)[B3]{0}
\obj(40,00)[B4]{1}
\obj(50,00)[B5]{1}

\mor{A1}{A2}{$\msf{add}(1)$}
\mor{A2}{A3}{$\msf{min}(0)$}
\mor{A3}{A4}{$(\msf{add}(1))$}[\atleft, \dotArrow]
\mor{A4}{A5}{$(\msf{min}(1))$}[\atleft, \dotArrow]
\mor{B1}{B2}{$\msf{add}(1)$}[\atright, \solidarrow]
\mor{B2}{B3}{$\msf{min}(0)$}[\atright, \solidarrow]
\mor{B3}{B4}{$(\msf{add}(1))$}[\atright, \dotArrow]
\mor{B4}{B5}{$(\msf{min}(1))$}[\atright, \dotArrow]

\mor{A2}{B4}{}
\mor{A3}{B5}{}
\mor{B2}{A4}{}
\mor{B3}{A5}{}
\enddc
\]
Here replicas $A$ and $B$ both increment the counter and then reset it, but the $\msf{add}(1)$ operations transform the concurrent $\min(0)$ operations into $\min(1)$ operations, giving a final state of 1.  Meanwhile, the observed-reset semantics would give a final state of 0.

Nonetheless, this resettable counter may still be useful, since the large state size of true observed-reset resettable counters \cite{resettable_counters} has led some to adopt alternative semantics in return for smaller state \cite{riak_datatypes, composable_embedded_counters}.
\end{itemize}
\end{myex}

\section{Related Work}
Several works describe general techniques for constructing replicated data types.  Leijnse, Almeida, and Baquero \cite{op_based_patterns} discuss patterns in existing op-based CRDTs, such as constructions of set CRDTs from flag CRDTs.  The patterns they identify are orthogonal to the semidirect product.

Baquero et~al.\ \cite{state_based_patterns} give general techniques for composing state-based CRDTs, the other kind of CRDTs besides op-based, using lattice merge functions.  Similar techniques are used by the Bloom$\mbox{}^L$ distributed programming model \cite{bloom_lattices} and the LVars parallel programming model \cite{lvars}.  While those works use lattice theory to compose state-based CRDTs, we use an idea from abstract algebra to compose op-based CRDTs.  Also, those works (especially \cite{bloom_lattices, lvars}) focus on composition in the sense of composite data types (e.g., tuples and maps), while we focus on composing different operations acting on the same base data type.

Mergeable Replicated Data Types (MRDTs), defined by Kaki et~al.\ \cite{mrdts} and built on top of Irmin \cite{irmin}, are branch-and-merge based replicated data types that use a three-way merge function for sets to define MRDTs for various data types automatically.  Like us, Kaki et~al.\ define a replicated integer register supporting addition and multiplication operations, but with different semantics: a multiplication is treated as its equivalent addition.  Our construction is less automatic but more flexible: MRDTs are only defined for data types built as views of relations on sets, and they give at most one semantics for a given data type, excluding examples like a remove-wins set or a reset-wins resettable type.  In contrast, our examples demonstrate the semidrect product's wide applicability.  Also, MRDT's branch-and-merge system model more closely resembles state-based CRDTs, in contrast to our use of op-based CRDTs.

A particular kind of arbitration between concurrent CRDT operations appears in work on tunable CRDTs by Rijo \cite{tunable_crdts} and the concept of ``cast-off updates'' in a survey by Pregui\c{c}a \cite{crdt_overview_preguica}.  Specifically, they consider operations that become irrevelant due to other operations, such as removes cancelled by adds in an add-wins set.  Rijo gives a construction based on generic arbitration rules, in which operations of one type can cancel those of another depending on their causal relationship.  

As described in Section \ref{sec:ot}, the semidirect product can be viewed as a restricted kind of operational transformation (OT).  We avoid the complexity of general OT approaches by effectively only allowing the transformation function to reorder operations, not arbitrarily transform them, and only in the specific case of a $\mc{C}_2$ operation followed by a $\mc{C}_1$ operation.  Also, we allow that the state is a CRDT state and the operations being transformed are CRDT messages, not just user operations acting on user-visible state, so that the components of a semidirect product may continue using CRDT techniques.

Lasp \cite{lasp} is an Erlang programming model that allows one to create views of a given CRDT, such as a functionally mapped view of a set, which update in an eventually consistent way.  We focus on adding in-place, mutating operations to CRDTs, as opposed to creating immutable views.  Indeed, several of our examples, such as the integer register with addition and multiplication operations and the map with a higher-order map operation, were motivated by the goal of developing CRDTs that support these operations in-place instead of as views.


OpSets \cite{opsets} and SECROs \cite{generic_rdt} both convert generic data types into replicated data types by sorting operations into an eventually consistent total order compatible with the causal order.  While we aim to sort operations so that they respect both the causal order and the arbitration order, this is not always possible (see Section \ref{sec:example}), so we instead use the transformation $\act$ to approximate the desired order.  As a result, the semidirect product gives semantics more typical of CRDTs than OpSets or SECROs: conflicts between concurrent operations are resolved uniformly according to an arbitration order chosen at design time, instead of according to an arbitrary total order at run time.  Also, through use of the transformation $\act$, we avoid the need to re-order and re-apply operations that are received out-of-order.

\section{Conclusion}
We introduced the semidirect product of op-based CRDTs.  This construction combines the operations of two CRDTs while handling concurrency conflicts between them in a uniform way.  Specifically, it implements an arbitration order on concurrent messages using a restricted kind of operational transformation.  We constructed novel CRDTs through composition, and our examples also showed that several existing CRDTs can be decomposed as semidirect products of simpler CRDTs.

For future work, we plan to investigate iterated semidirect products.  In particular, it would be interesting to see whether complicated CRDTs can be decomposed as iterated semidirect products of commutative data types, thus completely handling concurrency conflicts through the semidirect product.  We will also pursue implementations of the novel CRDTs described above.

\begin{acks}
We thank Carlos Baquero for feedback on an early draft of this work.
\end{acks}

\bibliographystyle{ACM-Reference-Format}
\bibliography{general_bib}

\appendix

\section{Algebraic Motivation}
\label{sec:algebra}
The semidirect product of CRDTs is inspired by the semidirect product of groups, which we now describe.

In abstract algebra, a \textit{group} is a set $G$ together with a binary operation $\bullet: G \times G \ra G$ such that:
\begin{itemize}
  \renewcommand\labelitemi{--}
  \item $\bullet$ is associative: $g \bullet (h \bullet k) = (g \bullet h) \bullet k$
  \item There is an identity $1_G \in G$ satisfying $1_G \bullet g = g \bullet 1_G = g$
  \item Each $g \in G$ has an inverse $g^{-1}$ such that $g \bullet g^{-1} = g^{-1} \bullet g = 1_G$.
\end{itemize}
Let $(G_1, \bullet_1)$ and $(G_2, \bullet_2)$ be groups.  Suppose we have an action $\act: G_2 \times G_1 \ra G_1$ satisfying:
\begin{itemize}
  \renewcommand\labelitemi{--}
  \item $g_2 \act (g_1 \bullet_1 g_1') = (g_2 \act g_1) \bullet_1 (g_2 \act g_1')$
  \item $(g_2 \bullet_2 g_2') \act g_1 = g_2 \act (g_2' \act g_1)$
  \item $1_{G_2} \act g_1 = g_1$
  \item For each $g_2$, $g_1 \mapsto g_2 \act g_1$ is an invertible function.
\end{itemize}
Then the \textit{semidirect product of $G_1$ and $G_2$ with respect to $\act$} is the group $G_1 \rtimes_\act G_2$ with underlying set $G_1 \times G_2$ and binary operation \cite[\S 5.5]{dummit_foote}
\[
(g_1, g_2) \bullet_\rtimes (g_1', g_2') := (g_1 \bullet_1 (g_2 \act g_1'), g_2 \bullet_2 g_2').
\]
We can think of $G_1$ and $G_2$ as subgroups of $G_1 \rtimes_\act G_2$ (i.e., subsets that are groups) via the maps $g_1 \mapsto (g_1, 1_{G_2})$, $g_2 \mapsto (1_{G_1}, g_2)$.  If we put the elements corresponding to $g_1$ and $g_2$ in the ``wrong'' order ($g_2$ then $g_1$), they get rearranged as
\[
(1_{G_1}, g_2) \bullet_\rtimes (g_1, 1_{G_2}) = (g_2 \act g_1, g_2) = (g_2 \act g_1, 1_{G_2}) \bullet_\rtimes (1_{G_1}, g_2),
\]
i.e., as ($g_2 \act g_1$ then $g_2$).  This inspired the semidirect product CRDT's use of the equivalence between $(m_2 \act m_1) \circ m_2$ and $m_2 \circ m_1$, for $m_1$ a $\mc{C}_1$ message and $m_2$ a $\mc{C}_2$ message, as a way to reorder messages so that $\mc{C}_1$ messages effectively come before $\mc{C}_2$ messages.


\end{document}